\documentclass[sigconf]{acmart}

\AtBeginDocument{%
  }

\setcopyright{acmlicensed}
\copyrightyear{2018}
\acmYear{2018}
\acmDOI{XXXXXXX.XXXXXXX}

\acmConference[CCS '24]{Make sure to enter the correct
  conference title from your rights confirmation emai}{October 14-18,
  2024}{ Salt Lake City, U.S.A.}
\acmISBN{978-1-4503-XXXX-X/18/06}

\begin{document}

\title{On the Relative Completeness of Satisfaction-based Probabilistic Hoare Logic With While Loop}

\author{Xin Sun }
\email{xin.sun.logic@gmail.com}
 
\affiliation{%
  \institution{  Zhejiang Lab}
  \city{Hangzhou}
  \state{Zhejiang}
  \country{China}
}

\author{Xingchi  Su }
\authornote{Corresponding author}
\email{x.su1993@gmail.com}
 
\affiliation{%
  \institution{ Zhejiang Lab}
  \city{Hangzhou}
  \state{Zhejiang}
  \country{China}
}

\author{Xiaoning Bian }
\email{bian@zhejianglab.com}
\affiliation{%
  \institution{  Zhejiang Lab}
  \city{Hangzhou}
  \state{Zhejiang}
  \country{China}
}

\author{Anran Cui}
\email{52265902013@stu.ecnu.edu.cn}

\affiliation{%
  \institution{East China Normal University}
  \city{Shanghai}
  \country{China}
}


\renewcommand{\shortauthors}{Not Given}

\begin{abstract}

Probabilistic Hoare logic (PHL) is an extension of Hoare logic and is specifically useful in verifying randomized programs. It allows researchers to formally reason about the behavior of programs with stochastic elements, ensuring the desired probabilistic properties are upheld. The relative completeness of satisfaction-based PHL has been an open problem ever since the birth of the first PHL in 1979. More specifically, no satisfaction-based PHL with While-loop has been proven to be relatively complete yet. This paper solves this problem by establishing a new PHL with While-loop and prove  its   relative completeness. The  programming language concerned in our PHL is expressively equivalent to the existing PHL systems but brings a lot of convenience in showing completeness. The weakest preterm for While-loop command reveals how it changes the probabilistic properties of computer states, considering both execution branches that halt and infinite runs. We prove the relative completeness of our PHL in two steps. We first establish a semantics and proof system of Hoare triples with probabilistic programs and deterministic assertions. Then, by utilizing the weakest precondition of deterministic assertions, we construct the weakest preterm calculus of probabilistic expressions. The relative completeness of our PHL is then obtained as a consequence of the weakest preterm calculus.
\end{abstract}



\maketitle

\section{Introduction}\label{sec.intro}

\textbf{Hoare Logic.} Hoare logic provides a formalization with logical rules on reasoning about the correctness of programs. It was originally designed by C. A. R. Hoare in 1969 in his seminal paper \cite{hoare1969axiomatic} which was in turn extended by himself in \cite{hoare1971procedures}. The underpinning idea captures the precondition and postcondition of executing a certain program. The precondition describes the property that the command relies on as a start. The postcondition describes the property that the command must lead to after each correct execution. Hoare logic has become one of the most influential tools in the formal verification of programs in the past decades. It has been successfully applied in analysis of deterministic \cite{hoare1969axiomatic,hoare1971procedures,Winskel93}, nondeterministic \cite{Dijkstra75,Dijkstra76,Apt84}, recursive \cite{Hoare71,FoleyH71,AptBO09}, probabilistic \cite{Ramshaw79,Hartog02,chadha2007reasoning,rand2015vphl}  and quantum  programs \cite{Ying11,LiuZWYLLYZ19,Unruh19,ZhouYY19,DengF22}. A comprehensive review of Hoare logic is referred to Apt, Boer, and Olderog \cite{apt2009verification,apt2019fifty}.


\noindent \textbf{Probabilistic Hoare Logic.} Probabilistic Hoare logic (PHL) \cite{Ramshaw79,Hartog02,chadha2007reasoning,rand2015vphl} is an extension of Hoare logic. It introduces probabilistic commands to handle programs with randomized behavior, providing tools to derive probabilistic assertions that guarantee a program fulfills its intended behavior with certain probabilities. Nowadays PHL plays important roles in the formal verification of cryptographic algorithm \cite{Hartog05,Hartog08,BartheGB09,BartheGB12,BartheDGKSS13}, machine learning algorithm \cite{SutskeverMDH13,SrivastavaHKSS14} and others systems involving uncertainty. 

Ramshaw \cite{Ramshaw79} developed the first Probabilistic Hoare Logic (PHL) using a truth-functional assertion language, where logic formulas are interpreted as either true or false. This type of PHL is called satisfaction-based PHL within the Hoare logic community. There are two types of formulas in this logic: deterministic formulas and probabilistic formulas. The truth value of deterministic formulas is interpreted on program states, which are functions that map program variables to their values. On the other hand, the truth value of probabilistic formulas is interpreted on the probability distribution of program states. However, Ramshaw's PHL is incomplete and may not be able to prove some simple and valid assertions.

To address this problem, expectation-based PHL was introduced in a series of work~\cite{Kozen85,Jones90,Morgan96, Morgan99}. This approach employs arithmetical assertions instead of truth-functional assertions. In this context, a Hoare triple $\{f\}C \{g\}$ represents that the expected value of the function $g$ after the execution of program $C$ should be at least as high as the expected value of the function $f$ before the execution.
 
\noindent \textbf{Different Probabilistic Commands.} Satisfaction-based PHL was developed further by den Hartog, Vink and Ricardo \cite{Hartog02,Hartog05,Hartog08}. Their PHL captures randomized behaviors by probabilistic choices, where the command $S_1$ is chosen with probability $\rho$ and the command $S_2$ is chosen with probability $1-\rho$, represented as $S_1\oplus_\rho S_2$. They also provide a denotational semantics accordingly and establish the completeness of the proof system without a while-loop. On the other hand, Chadha et al. \cite{chadha2007reasoning} constructed their PHL by incorporating randomness from tossing a biased coin. They showed that their PHL without the while-loop is complete and decidable. Rand and Zdancewic \cite{rand2015vphl} established the randomness of their PHL by also using a biased coin. They formally verified their logic in the Coq proof assistant.

\noindent \textbf{Our Contribution.} While recent work \cite{BatzKKM21} has proved that expectation-based PHL with the While loop is relatively complete, the work to date has not proven the relative completeness of any satisfaction-based PHL with the While loop. This is just the main contribution of this paper. To elaborate:

\begin{enumerate}
    \item We propose a new satisfaction-based PHL in which the randomness is introduced by the command of probabilistic assignment, {\em i.e.}, $X\xleftarrow{\$}\{a_1:k_{1},...,a_n:k_{n}\}$. This construction makes our logic concise in expressing random assignments with respect to discrete distribution, which are commonly seen in areas of cryptography, computer vision, coding theory and biology \cite{gordon2014probabilistic}. For example, in cryptographic algorithms, almost all nonces are chosen from some prepared discrete distributions on integers, rational or real numbers. Similarly, in the phase of parameter setting, a machine learning algorithm would choose parameters from a distribution over floating point numbers {\em w.r.t.} with some accuracy (discrete as well). The probabilistic assignment also brings a lot of convenience to the completeness proof since it can be treated as a probabilistic extension of the normal assignment. It is also expressively equivalent to the existing randomized commands, like probabilistic choices and biased coins. 

    \item We find out the appropriate weakest preterm for probabilistic expressions {\em w.r.t.} While-loop. It shows how While-loop changes the probabilistic properties of computer states, considering both execution branches that halt and infinite runs. As a preview, we prove the relative completeness of our PHL in two steps. We first establish a proof system of Hoare triples with deterministic assertions. Then, by utilizing the weakest precondition of deterministic assertions, we construct the weakest preterm calculus of probabilist expressions. The relative completeness of our PHL is then obtained as an application of the weakest preterm calculus.
\end{enumerate}

The outline of this paper is as follows. We first introduce our PHL with deterministic assertions in Section \ref{Probabilistic Hoare Logic with deterministic assertion}. We define the denotational semantics of deterministic assertions, construct a proof system and show that it is sound and relatively complete. Then Section \ref{Probabilistic Hoare Logic with probabilistic assertion} introduces the proof system for probabilistic assertions based on weakest preconditions and proves that it is relatively complete as well. We conclude this paper with future work in Section \ref{Conclusions and Future Work}.

\section{Probabilistic Hoare Logic with Deterministic Assertion}\label{Probabilistic Hoare Logic with deterministic assertion}

Hoare logic is a formal system that reasons about "Hoare triples" of the form $\{\phi\}C\{\psi\}$. A Hoare triple characterizes the effect of a command $C$ on the states that satisfy the precondition $\phi$, which means that if a program state satisfies $\phi$, it must also satisfy the postcondition $\psi$ after the correct execution of $C$ on the state. These assertions, also known as formulas, are built from deterministic and probabilistic expressions and will be defined in this section and the next. The commands $C$ are based on classical program statements such as assignment, conditional choice, while loop, and so on. This section will focus on the deterministic formulas.


\subsection{Deterministic Expressions and Formulas}
 
Let $\mathbb{PV}= \{X,Y,Z,\ldots\}$ be a set of program variables denoted by capital letters. Let $\mathbb{LV}= \{x,y,z,\ldots\}$ be a set of logical variables. We assume $\mathbb{LV}$ and $\mathbb{PV}$ are disjoint. Program variables are those variables that may occur in programs. They constitute deterministic expressions. 
Deterministic expressions are classified into arithmetic expression $E$ and Boolean expression $B$. The arithmetic expression consists of integer constant $n \in \mathbb{Z}$ and variables from $\mathbb{PV}$. It also involves arithmetic operators between these components. The arithmetic operator set is defined as $\{+,-,\times,...\}\subseteq \mathbb{Z}\times\mathbb{Z}\rightarrow \mathbb{Z}$.  In contrast, logical variables are used only in assertions. 

\begin{definition}[Arithmetic expressions]
     Given a set of program variables $\mathbb{PV}$, we  define the arithmetic expression $E$ as follows:

\begin{center}
    $E:=n\mid X\mid (E\ aop\ E)$.
\end{center}
\end{definition}

This syntax allows an arithmetic expression ($E$) to be either an integer constant ($n$), a program variable ($X$), or a composition of two arithmetic expressions ($E \ aop \ E$) built by an arithmetic operation ($aop$). They intuitively represent integers in programs. 

The Boolean constant set is $\mathbb{B}=\{\top,\bot\}$. We define relational operators ($rop$) to be performed on arithmetic expressions including $\{>,<,\geq,=,\leq,...\}\subseteq\mathbb{Z}\times\mathbb{Z}\rightarrow \mathbb{B}$. And logical operators, e.g., $\wedge,\vee,\neg,\rightarrow,...$, can be applied to any Boolean expressions.

\begin{definition}[Boolean expressions]
     
The Boolean expression is defined as follows:

 \begin{center}
     $B:=\top\mid \bot\mid (E\ rop\ E)\mid \neg B\mid (B\ lop\ B).$
 \end{center}
 
\end{definition}
 
A Boolean expression represents some truth value, true or false. The expression $(E\ rop\ E)$ represents that the truth value is determined by the binary relation $rop$ between two integers.

The semantics of deterministic expressions is defined on deterministic states $S$ which are denoted as mappings $S: \mathbb{PV}\rightarrow\mathbb{Z}$. Let $ \mathbb{S}$ be the set of all deterministic states. Each state $S \in  \mathbb{S}$ is a description of the value of every program variable. Accordingly, the semantics of arithmetic expressions is $[\![E]\!]:\ \mathbb{S}\rightarrow\mathbb{Z}$ which maps each deterministic state to an integer. Analogously, the semantics of Boolean expressions is $[\![B]\!]:\ \mathbb{S}\rightarrow\mathbb{B}$ which maps each state to a Boolean value. 

\begin{definition}[Semantics of deterministic expressions]\label{semanticsdeterministicexpression}

The semantics of deterministic expressions are defined inductively as follows:

\begin{center}
    \begin{tabular}{rll}
      $[\![X]\!]S$ & = & $S(X)$ \\
      $[\![n]\!]S$ & = & $n$ \\
      $[\![E_{1}\ aop\ E_{2}]\!]S$ & = & $[\![E_{1}]\!]S\ aop\ [\![E_{2}]\!]S$\\
      $[\![\top]\!]S$ & = & $\top$\\
      $[\![\bot]\!]S$ & = & $\bot$\\
      $[\![E_{1}\ rop\ E_{2}]\!]S$ & = & $[\![E_{1}]\!]S\ rop\ [\![E_{2}]\!]S$\\
      $[\![\neg B]\!]S$ & = & $\neg[\![B]\!]S$\\
      $[\![B_{1}\ lop\ B_{2}]\!]S$ & = & $[\![B_{1}]\!]S\ lop\ [\![B_{2}]\!]S$
    \end{tabular}
\end{center}
\end{definition}

As mentioned above, the interpretation of an arithmetic expression is an integer. A program variable $X$ on a deterministic state is interpreted as its value on the state. A constant is always itself over any state. An arithmetic expression $E_1\ aop\ E_2$ is mapped to the integer calculated by the operator $aop$ applied on the interpretation of $E_1$ and the interpretation of $E_2$ on the state. The Boolean expressions can be understood similarly. For example, let $S$ be a state such that $S(X)=1$ and $S(Y)=2$. Then $[\![X\ +\ 1]\!]S=2$ and $[\![(X+2\leq 3) \wedge\ (X+Y=3)  ]\!]S=\top$. 

Next we  define deterministic formulas based on deterministic expressions. 

\begin{definition}[Syntax of deterministic formulas]\label{def.syntaxdeterministicformulas}
The deterministic formulas are defined by the following BNF:
\begin{center}
$\phi:= \top\mid \bot\mid (e\ rop\ e)\mid \neg\phi\mid (\phi\ lop\ \phi)\mid \forall x\phi$
\end{center}
where $e$ represents arithmetic expression build on $\mathbb{LV} \cup \mathbb{PV}$:

\begin{center}
    $e:=n\mid X\mid x\mid (e\ aop\ e)$.
\end{center}

\end{definition}

We restricts $lop$ to the classical operators: $\neg$ and $\wedge$. $\vee$ and $\to$ can be expressed in the standard way. The formula $\forall x\phi$ applies universal quantifier  to the logical variable $x$ in formula $\phi$. 

An interpretation $I : \mathbb{LV} \mapsto \mathbb{Z}$ is a function which maps logical variables to integers. Given an interpretation $I$ and a deterministic state $S$, the semantics of $e$ is defined as follows.

\begin{center}
    \begin{tabular}{rll}
     $[\![n]\!]^IS$ & = & $n$ \\
      $[\![X]\!]^IS$ & = & $S(X)$ \\
      $[\![x]\!]^IS$ & = & $I(x)$ \\
      $[\![E_{1}\ aop\ E_{2}]\!]^IS$ & = & $[\![E_{1}]\!]^IS\ aop\ [\![E_{2}]\!]^IS$\\
    \end{tabular}
\end{center}

The semantics of a deterministic formula is denoted by $[\![\phi]\!]^I=\{S\mid S\models^I\phi\}$ which represents the set of all states satisfying $\phi$. 

\begin{definition}[Semantics of deterministic formulas]\label{def.semanticsdeterministicformula}
The semantics of deterministic formulas is defined inductively as follows:

\begin{center}
\begin{tabular}{p{.1\textwidth}p{.02\textwidth}p{.3\textwidth}}
$[\![\top]\!]^I$ & = & $\mathbb{S}$\\
$[\![\bot]\!]^I$ & =& $\emptyset$\\
$[\![e_{1}\ rop\ e_{2}]\!]^I$ & = & $\{S\in\mathbb{S}\mid [\![e_{1}]\!]^I S\ rop\ [\![e_{2}]\!]^I S=\top\}$\\
$[\![\neg\phi]\!]^I$ & = & $\mathbb{S}\backslash[\![\phi]\!]^I$\\
$[\![\phi_{1}\wedge\phi_{2}]\!]^I$ & = & $[\![\phi_{1}]\!]^I\cap[\![\phi_{2}]\!]^I$\\
$[\![\forall x\phi]\!]$ &  $=$ & $\{ S\mid$ for all integer $n$ and $I'=I[x \mapsto n]$, $S\models^{I'} \phi\}$
        \end{tabular}
    \end{center}
\end{definition}

The $[\![\top]\!]^I$ defaults to all deterministic states $\mathbb{S}$, while $[\![\bot]\!]^I$ is interpreted as the empty set. The symbol $\backslash$ denotes complement, and $[\![\neg\phi]\!]^I$ represents the set of remaining states in $\mathbb{S}$ after removing all states satisfying $\phi$. The logical operations $\wedge$ and $\vee$ between formulas can be interpreted as intersection and union operation of state sets which satisfy corresponding formulas, respectively. And the formula $\forall x \phi$ is satisfied on a deterministic state with interpretation $I$ if and only if $\phi$ is true with respect to all interpretations $I'$ which assigns the same values to every variable as $I$ except $x$. 

For example, let $S$  be a state such that $S(X)=1$ and let $I(x)=3$.  The deterministic formula $\forall x((x>0)\to (x+X>X))$ is satisfied on $S$, i.e. $S\models^I \forall x(x>0\to x+X>X)$. It is also valid (satisfied on arbitrary state and interpretation).

\subsection{Commands}

Commands are actions that we perform on program states. They change a deterministic state to a probabilistic distribution of deterministic states.	We introduce the probabilistic assignment command to capture the randomized executions of probabilistic programs.

\begin{definition}[Syntax of command expressions]\label{syntaxcommandexpression}
    The commands are defined inductively as follows: 
   
\begin{center}
    $C:=\texttt{skip}\mid X\leftarrow E\mid X\xleftarrow[]{\$}R\mid C_{1};C_{2}\mid \texttt{if}\ B\ \texttt{then}\ C_{1}\ \texttt{else}\ C_{2}  \mid \texttt{while} \mbox{ } B \mbox{ } \texttt{do} \mbox{ } C $
\end{center}

\noindent
where $R=\{a_1: k_1,\cdots,a_n: k_n\}$ in which  $\{ k_1,\cdots, k_n\}$ is a set of integers and $a_1, \ldots , a_n$ are real numbers such that $0\leq a_i \leq 1$ and $a_1+ \ldots+a_n =1$. We omit those $a_i$s when they are all equal to $\frac{1}{n}$. $B$ is deterministic formula. 
\end{definition}

The command $\texttt{skip}$ represents a null command doing nothing. $X\leftarrow E$ is the deterministic assignment. $C_1;C_2$ is the sequential composition of $C_1$ and $C_2$ as usual. The last two expressions are the conditional choice and loop, respectively. $X\xleftarrow[]{\$}R$ can be read as a value $k_i$ is chosen with probability $a_i$ and is assigned to $X$. The probabilistic assignment is the way to introduce randomness in this paper. It is worth noting that the language we use for command expressions is just as expressive as the languages that are constructed by using biased coins or probabilistic choices. This can be easily understood through an example: a probabilistic choice $C_1\oplus_{\frac{1}{3}} C_2$ is equivalent to the following program in our language (assuming that $X$ is a new program variable): 

$$X\xleftarrow[]{\$}\{\frac{1}{3}:0,\frac{2}{3}:1\}; \texttt{if}\ (X=0)\ \texttt{then}\ C_1\ \texttt{else}\ C_2$$

 
The semantics of commands is defined on probabilistic states. It shows how different commands \textit{update} probabilistic states.  A probabilistic state, denoted by $\mu$, is a probability sub-distribution on deterministic states, {\em i.e.}, $\mu\in D(\mathbb{S})$. Thus, each $\mu: \mathbb{S}\to [0,1]$ requires that $\Sigma_{S\in\mathbb{S}} \mu(S)\leq 1$. We use sub-distributions to take into account the situations where some programs may never terminate in certain states. For a deterministic state $S\in\mathbb{S}$, $\mu_{S}$ is a special probabilistic state that assigns the value of 1 to $S$ and the value of 0 to any other state. We call it the probabilistic form of a deterministic state. A deterministic state $S$ is considered to be a support of $\mu$ if $\mu(S)>0$. The set of all supports of $\mu$ is denoted by $sp(\mu)$.


\begin{definition}[Semantics of command expressions]\label{def.semanticsofcommand}
The semantics of commands is a function $[\![C]\!]\in  D(\mathbb{S})\rightarrow D(\mathbb{S})$. It is defined inductively as follows:

\begin{itemize}

\item  $[\![  \texttt{skip}  ]\!](\mu)=\mu$

\item $[\![X\leftarrow E]\!](\mu)=    \displaystyle\sum\limits_{S \in \mathbb{S}}  \mu(S)\cdot  \mu_{S[X\mapsto[\![E]\!]S]}$

\item $[\![X\xleftarrow{\$}\{a_1:k_{1},...,a_n:k_{n}\}]\!](\mu)=  \displaystyle\sum_{i=1}^{n} a_i [\![X\leftarrow k_{i}]\!](\mu)$ 

\item $[\![C_{1};\ C_{2}]\!](\mu)=[\![C_{2}]\!]([\![C_{1}]\!](\mu))$

\item   $[\![\texttt{if}\ B\ \texttt{then}\ C_{1}\ \texttt{else}\ C_{2}]\!](\mu)=[\![C_{1}]\!](\downarrow_{B}(\mu))+[\![C_{2}]\!](\downarrow_{\neg B}(\mu))$

\item $[\![   \texttt{while} \mbox{ } B \mbox{ } \texttt{do} \mbox{ } C  ]\!] (\mu ) =  \displaystyle\sum_{i = 0}^{\infty} \downarrow_{\neg B}  (  (  [\![   C   ]\!] \circ \downarrow_{ B})^i (\mu ) ) $

\end{itemize}

\end{definition}


The command $\texttt{skip}$ changes nothing. We write $S[X\mapsto[\![E]\!]S]$ to denote the state which assigns variables the same values as $S$ except that the variable $X$ is assigned the value $[\![E]\!]S$.  Here $\downarrow_{B}(\mu)$ denotes the distribution $\mu$ restricted to those states where $B$ is true. Formally, $\downarrow_{B}(\mu)=v$ with $v(S)=\mu(S)$ if $[\![B]\!]S=\top$ and $v(S)=0$ otherwise. We can write $[\![C]\!]S$ to denote $[\![C]\!]\mu_S$ if the initial state is deterministic. 

In general, if $\mu = [\![   C  ]\!] S $,  $S' \in sp (\mu)$ and $\mu (S') = a$. Then it means that executing command $C$ from state $S$ will terminate on state $S'$ with probability $a$. 


\begin{example}
    Let $R=\{\frac{1}{2} :0, \frac{1}{2} :1\}$ and let $S$ be a deterministic state such that $S(X) = 1$. If we run the command $X\xleftarrow{\$}R$ on  $S$, then distribution $[\![X\xleftarrow{\$}\{\frac{1}{2} :0, \frac{1}{2} :1\} ]\!]\mu_{S}=\frac{1}{2}(\mu_{S[X\mapsto 0]})+\frac{1}{2}(\mu_{S[X\mapsto 1]})$ is obtained.
\end{example}

\begin{example} Let $\textbf{0}$  be the probabilistic state that maps every deterministic state to 0.
    For any probabilistic state $\mu $,     $$[\![   \texttt{while} \mbox{ } \top \mbox{ } \texttt{do} \mbox{ } \texttt{skip}  ]\!] (\mu  ) = \textbf{ 0}.$$  This is because $$[\![   \texttt{while} \mbox{ } \top \mbox{ } \texttt{do} \mbox{ } \texttt{skip}  ]\!] (\mu  ) = $$ 
     $$     \displaystyle\sum_{i = 0}^{\infty} \downarrow_{\neg ( \top )}  (  (  [\![     \texttt{skip}   ]\!] \circ \downarrow_{ \top })^i (\mu ) )  = $$ 
     $$  \downarrow_{\neg (\top)} (\mu) +    \downarrow_{\neg ( \top)}  (  (  [\![     \texttt{skip}   ]\!]  \circ \downarrow_{  \top })  (\mu ) )  +  \downarrow_{\neg ( \top )}  (  (  [\![     \texttt{skip}   ]\!] \circ \downarrow_{ \top })^2 (\mu ) ) + \ldots  $$ It's easy to see that $ \downarrow_{\neg ( \top )} (\mu) =  \textbf{ 0}$ and $  (  [\![     \texttt{skip}   ]\!] \circ \downarrow_{ \top })^k (\mu )  = \mu$ for all $k$. Therefore, $[\![   \texttt{while} \mbox{ } \top  \mbox{ } \texttt{do} \mbox{ } \texttt{skip}  ]\!] (\mu  ) = \textbf{ 0}$. The statement implies that certain programs that never terminate result in probabilistic states \textbf{0}.
    
\end{example}

\begin{example}   Assume that there are two  variables $X,Y$ and infinitely many states $S_0 , S_1, \ldots$ where   $S_0(X)= 0, S_0 (Y) = 0$,  $S_i(X) =1, S_i (Y) =i$ for all $i>0$.   
   Consider the command    $ C:=   \texttt{while} \mbox{ } X=0 \mbox{ } \texttt{do} \mbox{ }   ( X \leftarrow \{\frac{1}{2} :0, \frac{1}{2} :1\}  ; Y\leftarrow Y+1 )  $. If we let   $\mu=[\![   C  ]\!] (S_0)  $, then $\mu(S_0) =0$ and $\mu(S_i) = \frac{1}{2^i}$ for all $i>0$.

\end{example}




Definition~\ref{def.semanticsdeterministicformula} gives the semantics of deterministic formulas over deterministic states. A deterministic formula describes some property of deterministic states. But how to evaluate a deterministic formula on probabilistic states? The semantics is given as follows:  

\begin{center}
    $\mu\models^I \phi$   iff   for each support $S$ of $\mu$, $S\in [\![\phi]\!]^I$.
\end{center}
 
\noindent We call it possibility semantics because the definition intuitively means that a formula $\phi$ is true on a probabilistic state if and only if $\phi$ is true on all possible deterministic states indicated by the probabilistic state. That implies that all supports of a distribution share a common property. Hence we can claim that the distribution satisfies the formula. The possibility semantics makes our PHL with deterministic formula (PHL$_d$)  essentially equivalent to Dijkstra's non-deterministic Hoare logic \cite{Dijkstra75}. However, the former serves as a better intermediate step towards PHL with probabilistic formulas than the latter. Therefore, we will still present  PHL$_d$ in detail, especially the weakest precondtion calculus of PHL$_d$, which is not concretely introduced in non-deterministic Hoare logic.

\subsection{Proof System with deterministic assertions}


A proof system for PHL is comprised of Hoare triples. A Hoare triple, written as  $\{\phi\}C\{\psi\}$, is considered valid if, for every deterministic state that satisfies $\phi$, executing command C results in a probabilistic state that satisfies $\psi$. Formally,  

\begin{center}
$\models \{\phi\}C\{\psi\}  $ if for all interpretation $I$ and deterministic state $S$, if $S\models^I \phi$, then $[\![C]\!](\mu_{S}) \models^I \psi$.
\end{center}


We now build a proof system for PHL$_d$ for the derivation of Hoare triples with probabilistic commands and deterministic assertions. Most rules in our proof system are standard, and they are inherited from Hoare logic or natural deduction \cite{apt2019fifty}. Only one new rule for probabilistic assignment is added, along with some structural rules.  The symbol $\phi[X/E]$ represents the formula which replaces every occurrence of $X$ in $\phi$ with $E$.

\begin{definition}[Proof system of PHL$_d$]
    The proof system of PHL$_d$ consists of the following inference rules:

    \begin{center}
        \begin{tabular}{ll}
         $SKIP:$ & $\frac{}{\vdash\{\phi\}\texttt{skip}\{\phi\}}$\\
         $AS:$ & $ \frac{}{\vdash \{   \phi[X/ E]   \}  X \leftarrow E\{ \phi  \} }$\\

         $PAS:$ & $   \frac{}{\vdash \{   \phi[X/k_1] \wedge \ldots \wedge  \phi[X/k_n]   \}  X\xleftarrow{\$}\{a_1: k_{1},...,a_n: k_{n} \}   \{ \phi  \} }$\\  
         \\
         $SEQ:$ & $ \frac{\vdash\{\phi\}C_{1}\{\phi_{1}\}\quad \vdash\{\phi_{1}\}C_{2}\{\phi_{2}\}} {\vdash\{\phi\}C_{1};C_{2}\{\phi_{2}\}} $\\
        \\
          $IF:$ & $ \frac{\vdash\{\phi\wedge B\}C_{1}\{\psi\}\quad \vdash\{\phi\wedge\neg B\}C_{2}\{\psi\}} {\vdash\{\phi\} \texttt{if}\ B\ \texttt{then}\ C_{1}\ \texttt{else}\ C_{2}\{\psi \} } $\\
          \\
          $CONS:$ & $ \frac{\models\phi'\rightarrow\phi\quad  \vdash\{\phi\}C\{\psi\}\quad \models\psi\rightarrow\psi'}  {\vdash\{\phi'\}C\{\psi'\}}$\\
          \\
          $AND:$ & $ \frac{\vdash\{\phi_{1}\}C\{\psi_{1}\}\quad \vdash\{\phi_{2}\}C\{\psi_{2}\}} {\vdash\{\phi_{1}\wedge\phi_{2}\}C\{\psi_{1}\wedge\psi_{2}\}}$\\
          \\
          $OR:$ & $ \frac{\vdash\{\phi_{1}\}C\{\psi_1\}\quad \vdash\{\phi_{2}\}C\{\psi_2\}} {\vdash\{\phi_{1}\vee\phi_{2}\}C\{\psi_1 \vee \psi_2\}}$\\
          \\
           $WHILE:$ & $\frac{\vdash \{\phi  \wedge B\}   C   \{\phi\}    } {\vdash\{\phi  \}    \texttt{while} \mbox{ }B \mbox{ } \texttt{do} \mbox{ }C   \{\phi \wedge \neg B    \}}$
        \end{tabular}
    \end{center}
\end{definition}

The majority of the above inference rules are easy to comprehend. $(CONS)$ is special since it involves semantically valid implications in the premise part. It characterizes the monotonicity of Hoare triples, which means that a stronger precondition must also lead to the same postcondition or some weaker one. In rule $(WHILE)$, formula $\phi$ is called loop invariant which will not be changed by command $C$. In the remaining part of this section, we prove the soundness and completeness of PHL$_d$. Most of the proofs are similar to their analogue in classical Hoare logic, the confident readers may feel free to skip them.



\begin{theorem}[Soundness]\label{th.soundnessphl}
 For all deterministic formula $\phi$ and $\psi$ and command $C$, $\vdash\{\phi\}C\{\psi\}$ implies $\models\{\phi\}C\{\psi\}$. 
\end{theorem}

\begin{proof}

We prove by structural induction on $C$.  Let $I$ be an arbitrary interpretation.
\begin{itemize}
\item (SKIP) It's trivial to see that that $\models \{\phi\}\texttt{skip}\{\phi\}$.

\item (AS) Assume $ S \models^I \phi[X/ E ] $. This means that $\phi$ is true if the variable $X$ is assigned to the value $[\![E]\!]S$ and all other values are assigned to a value according to $S$. Let $S' = [\![X\leftarrow E]\!](S)=  S[X\mapsto[\![E]\!]S]$. Then $S'$ assigns $X$ to the value $[\![E]\!]S$ and all other variables to the same value as $S$. Therefore, $S' \models^I \phi$.

\item (PAS) Assume $ S \models^I \phi[X/k_1] \wedge \ldots \wedge  \phi[X/k_n] $. This means that $\phi$ is true if the variable $X$ is assigned to the any of $\{k_1, \ldots, k_n\}$ and all other values are assigned to a value according to $S$. Let $\mu' = [\![   X\xleftarrow{\$}\{a_1: k_{1},...,a_n: k_{n}\}     ]\!](\mu_S)$. Then $\mu'$ is a distribution with support $\{S_1, \ldots, S_n\}$, where $S_i =  S[X\mapsto k_i]$ for $i\in \{1, \ldots,n\}$. Since $S_i$ assigns $X$ to the value $k_i$ and all other variables to the same value as $S$. We know that $S_i \models^I \phi$. This means that $\phi$ is true on all supports of $\mu'$. Therefore, $\mu' \models^I \phi$.

\item (SEQ)  If rule (SEQ) is used to derive $\vdash\{\phi\}C_{1};C_{2}\{\phi_{2}\}$ from $\vdash\{\phi\}C_{1}\{\phi_{1}\}$ and $\vdash\{\phi_{1}\}C_{2}\{\phi_{2}\}$, then by induction hypothesis we have $\models\{\phi\}C_{1}\{\phi_{1}\}$ and $\models\{\phi_{1}\}C_{2}\{\phi_{2}\}$. Assume $S\in [\![\phi]\!]^I$. Let $S' \in sp(   [\![C_{1}; C_{2}]\!](\mu_{S})  )$ be an arbitrary state which belongs to the support of $[\![C_{1}; C_{2}]\!](\mu_{S}) $. From $[\![C_{1}; C_{2}]\!](\mu_{S})  =    [\![C_{2}]\!]([\![C_{1}]\!](\mu_{S}))   $ we know that there is a state $S_1 \in sp(   [\![C_{1}]\!](\mu_{S})  )$ such that $S' \in sp (   [\![C_{2}]\!](\mu_{S_1})   )$. Now by  $\models  \{ \phi \} C_{1} \{ \phi_{1} \}$   we know that $ [\![C_{1}]\!](\mu_{S})  \models^I \phi_1$ and  $S_1 \models^I \phi_1$. By $\models \{ \phi_{1} \}C_{2} \{ \phi_{2} \}$ we know that $ [\![C_{2}]\!](\mu_{S_1})  \models^I \phi_2$ and  $S'\models^I \phi_2$.

\item (IF)   Assume $S\in [\![\phi]\!]^I$, $\vdash\{\phi\wedge B\}C_{1}\{\psi\}$ and $\vdash\{\phi\wedge \neg B\}C_{2}\{\psi\}$. By induction hypothesis we know that $ \models  \{\phi\wedge B\}C_{1}\{\psi\}$ and $\models  \{\phi\wedge \neg B\}C_{2}\{\psi\}$. Let $S'$ be an arbitrary state which belongs to $sp([\![\texttt{if}\ B\ \texttt{then}\ C_{1}\ \texttt{else}\ C_{2}]\!](\mu_{S}))$.

Since $S$ is a deterministic state, we know that either $ S \in  [\![B]\!] $ or $ S \in  [\![\neg B]\!] $.

\begin{itemize}

\item If $ S \in  [\![B]\!] $, then  $[\![\texttt{if}\ B\ \texttt{then}\ C_{1}\ \texttt{else}\ C_{2}]\!](\mu_{S})=[\![C_{1}]\!](\downarrow_{B}(\mu_S))  =   [\![C_{1}]\!]( \mu_S)  $. Hence $S'\in  sp ( [\![C_{1}]\!]( \mu_S) )$. From $S\in [\![\phi]\!]^I$ and $ S \in  [\![B]\!] $ we know that $S \models^I  \phi \wedge B$. Now by $ \models  \{\phi\wedge B\}C_{1}\{\psi\}$ we deduce that $[\![C_{1}]\!]( \mu_S) \models^I \psi$. Therefore, $S' \models^I  \psi$.

\item If $ S \in  [\![ \neg B]\!] $, then  $[\![\texttt{if}\ B\ \texttt{then}\ C_{1}\ \texttt{else}\ C_{2}]\!](\mu_{S})=[\![C_{2}]\!](\downarrow_{ \neg B}(\mu_S))  =   [\![C_{2}]\!]( \mu_S)  $. Hence $S'\in sp([\![C_{2}]\!]( \mu_S) )$. From $S\in [\![\phi]\!]^I$ and $ S \in  [\![\neg B]\!] $ we know that $S \models^I  \phi \wedge \neg B$. Now by $ \models  \{\phi\wedge \neg B\}C_{2}\{\psi\}$ we deduce that $[\![C_{2}]\!]( \mu_S) \models^I \psi$. Therefore, $S' \models^I  \psi$.
\end{itemize}

\item (CONS) Assume $\models \phi'\rightarrow \phi$, $\vdash \{\phi\} C \{\psi\} $ and $\models \psi \rightarrow \psi'$. By induction hypothesis we obtain that $\models \{\phi\} C \{\psi\}$. Let $S$ be a state such that $S \models^I  \phi'$. Then by $\models \phi'\rightarrow \phi$ we know that $S \models^I  \phi$. By $\models \{\phi\} C \{\psi\}$ we know that $ [\![C]\!] (S) \models^I \psi $. Hence $S' \models^I \psi$ for all $S'$ which belongs to $sp([\![C]\!] (S))$. Now by  $\models \psi \rightarrow \psi'$ we know that $S' \models^I \psi'$.

\item (AND) Assume $\vdash\{\phi_{1}\}C\{\psi_{1}\} $ and $  \vdash\{\phi_{2}\}C\{\psi_{2}\}$. By induction hypothesis we know that $\models \{\phi_{1}\}C\{\psi_{1}\} $ and $  \models \{\phi_{2}\}C\{\psi_{2}\}$. Let $S$ be a state such that $S \models^I  \phi_1 \wedge \phi_2$.  Let $S'$ be a state in $sp([\![C]\!] (S))$.
Then $S \models^I \phi_1$ and by  $\models \{\phi_{1}\}C\{\psi_{1}\} $ we have $[\![C]\!] (S)  \models^I \psi_1$.  Hence $S' \models^I \psi_1$. Similarly, we can deduce that $S' \models^I \psi_2$. Therefore, $S' \models^I \psi_1 \wedge \psi_2$.

\item (OR) Assume $\vdash\{\phi_{1}\}C\{\psi_{1}\} $ and $  \vdash\{\phi_{2}\}C\{\psi_{2}\}$. By induction hypothesis we know that $\models \{\phi_{1}\}C\{\psi_{1}\} $ and $\models \{\phi_{2}\}C\{\psi_{2}\}$. Let $S$ be a state such that $S \models^I  \phi_1 \vee \phi_2$.  Let $S'$ be a state in $sp([\![C]\!] (S))$.
Then either  $S \models^I \phi_1$ or $S  \models^I  \phi_2$. If $S \models^I \phi_1$, then by $\models \{\phi_{1}\}C\{\psi_{1}\} $ we have $[\![C]\!] (S)  \models^I \psi_1$.  Hence $S' \models^I \psi_1$ and $S' \models^I \psi_1 \vee \psi_2$.  If $S \models^I \phi_2$, then by $\models \{\phi_{2}\}C\{\psi_{2}\} $ we have $[\![C]\!] (S)  \models^I \psi_2$.  Hence $S' \models^I \psi_2$ and $S' \models^I \psi_1 \vee \psi_2$.  
Therefore, it holds that $S' \models^I \psi_1 \vee \psi_2$ no matter $S \models^I \phi_1$ or $S \models^I \phi_2$.

\item (WHILE) Assume $\models \{\phi  \wedge B\}   C   \{\phi\} $. Let $S$ be an arbitrary state and $I$ be an arbitrary interpretation such that $S\models^I \phi$. We remind the readers that $ [\![   \texttt{while} \mbox{ } B \mbox{ } \texttt{do} \mbox{ } C  ]\!] (\mu_S ) =  \displaystyle\sum_{i = 0}^{\infty} \downarrow_{\neg B}  (  (  [\![   C   ]\!] \circ \downarrow_{ B})^i (\mu_S ) )$. We claim that for each natural number $k$, it holds that  $     (  [\![   C   ]\!] \circ \downarrow_{ B})^k (\mu_S ) \models^I \phi  $.

We prove by induction on $k$:

\begin{itemize}
\item If $k=0$, then $([\![C]\!] \circ \downarrow_{ B})^k (\mu_S )  =        \mu_S   $.  Hence $  (  [\![   C   ]\!] \circ \downarrow_{ B})^k (\mu_S )  \models^I  \phi$.

\item Assume it holds that $     (  [\![   C   ]\!] \circ \downarrow_{ B})^k (\mu_S ) \models^I \phi  $.
Then $     (  [\![   C   ]\!] \circ \downarrow_{ B})^{k+1} (\mu_S )  =      [\![   C   ]\!] \circ \downarrow_{ B}  ( (  [\![   C   ]\!] \circ \downarrow_{ B})^{k} (\mu_S )   )$. For arbitrary $S' \in sp (   (  [\![   C   ]\!] \circ \downarrow_{ B})^k (\mu_S )   )$, we have $S' \models^I \phi$. 
\begin{itemize}
    \item If $S' \models \neg B$, then $sp([\![   C   ]\!] \circ \downarrow_{ B} ( \mu_{S'})) =\emptyset$. Then $[\![   C   ]\!] \circ \downarrow_{ B} ( \mu_{S'}) \models^I \phi$ vacuously.
    \item If $S' \models  B$, then $[\![   C   ]\!] \circ \downarrow_{ B} ( \mu_{S'})= [\![   C   ]\!]   ( \mu_{S'})$. Then by $\models \{\phi  \wedge B\}   C   \{\phi\} $ we know $[\![   C   ]\!]  ( \mu_{S'}) \models^I \phi$. 
\end{itemize}
Therefore, we know $[\![   C   ]\!] \circ \downarrow_{ B} ( \mu_{S'}) \models^I \phi$ for all  $S' \in sp (   (  [\![   C   ]\!] \circ \downarrow_{ B})^k (\mu_S )   )$. This implies that $ (  [\![   C   ]\!] \circ \downarrow_{ B})^{k+1} (\mu_S ) \models^I \phi$.
\end{itemize}

With the above claim we infer that $\downarrow_{\neg B}  (  (  [\![   C   ]\!] \circ \downarrow_{ B})^i (\mu_S ) ) \models^I \phi \wedge \neg B$. 
	
Hence $\displaystyle\sum_{i = 0}^{\infty} \downarrow_{\neg B}  (  (  [\![   C   ]\!] \circ \downarrow_{ B})^i (\mu_S ) ) \models^I \phi  \wedge \neg B$. This proves $  [\![   \texttt{while} \mbox{ } B \mbox{ } \texttt{do} \mbox{ } C  ]\!] (\mu_S ) \models^I  \phi  \wedge \neg B $.   
\end{itemize}    

\end{proof}

\vspace{-5mm}

\subsection{Weakest precondition of deterministic assertions}

To give the completeness of PHL$_d$, we assume that the set of program variables $\mathbb{PV}$ is finite. Under this assumption, every deterministic state can be characterized by a unique formula, respectively. For example,  $S=\{X\mapsto 1, Y\mapsto 2, Z\mapsto 3\}$ is characterized by the formula $(X=1\wedge Y=2\wedge Z=3)$ when $\mathbb{PV}=\{X,Y,Z\}$. 

The completeness is given in the light of the method of weakest preconditions. Let $wp(C,\psi)$ denote the weakest precondition of a command $C$ and a postcondition $\psi$. Intuitively, $[\![ wp(C,\psi) ] \!]$ is the largest set of states starting from which if a program $C$ is executed, the resulting states satisfy a given post-condition  $\psi$. $[\![ wp(C,\psi) ] \!]$ is a precondition in the sense that $\models \{wp(C,\psi) \} C \{ \psi\}$ and it is weakest since for each formula $\phi$, if $\models \{\phi\}C\{\psi\}$, then $\models \phi \rightarrow wp(C,\psi)$. In the strict sense, the weakest precondition is a semantic notion. That is, $[\![ wp(C,\psi)] \!]$ is the weakest precondition of command $C$ and  postcondition $\psi$, while $ wp(C,\psi)$ is a syntactic representative of the weakest precondition. In this article, however, we will treat weakest precondition as a syntactic notion when this treatment is more convenient for our constructions and proofs.

The proof strategy can be performed step by step:

\begin{itemize}
    \item Step 1: Find out the corresponding precondition for each command;
    \item Step 2: Prove that $\{   wp(C, \phi) \}   C\{\phi\}$ is derivable by the proof system;
    \item Step 3: Prove that for each valid Hoare triple $\{\phi\}C\{\psi\}$, it holds that $\phi\to wp(C,\psi)$ is valid;
    \item Step 4: By inference rule $(CONS)$, it is obtained that $\{\phi\}C\{\psi\}$ is derivable by the proof system. 
\end{itemize}

The subtleties in this proof are to find the proper definition of the weakest preconditions with respect to different commands.

\begin{definition}[Weakest precondions]\label{def.weakestprecondition}
The weakest precondition is defined inductively on the structure of commands as follows:

\begin{enumerate}
\item $wp(skip, \phi) = \phi$.
\item  $wp(  X \leftarrow E, \phi )= \phi[X/ E ]$.
\item  $wp( X\xleftarrow{\$}\{a_1: k_{1},...,a_n: k_{n} \}, \phi  ) =   \phi[X/k_1] \wedge \ldots \wedge  \phi[X/k_n]$.
\item $wp( C_1 ; C_2 , \phi) = wp(C_1, wp(C_2, \phi))$.
\item $wp( if\ B\ then\ C_{1}\ else\ C_{2}, \phi   )=  (B \wedge wp(C_1, \phi))  \vee (\neg B \wedge wp(C_2, \phi))$.
\item $wp( while \mbox{ }B \mbox{ } do \mbox{ }C  , \phi)= \bigwedge\limits_{k \geq 0}  \psi_k   $, 
where $\psi_0 = \top$ and $\psi_{i+1} = ( B \wedge wp( C, \psi_{i}))  \vee ( \neg B  \wedge \phi)$.
\end{enumerate}
\end{definition}

Some careful readers may notice that $\bigwedge\limits_{k \geq 0}  \psi_k $, albeit understandable, is not an assertion in our language because it uses infinite conjunction. In fact, just like in classical Hoare logic,  $wp( while \mbox{ }B \mbox{ } do \mbox{ }C  , \phi)$ can be equivalently expressed as an assertion by using G$\ddot{o}$del's $\beta$ predicate. For example, in the special case when there is only one program variable $X$ mentioned in $C$ and $\phi$, we have 

$$ wp( while \mbox{ }B \mbox{ } do \mbox{ }C  , \phi)=$$ 
$$\forall k \forall m, n \geq 0 ( (  \beta^{\pm}(n,m,0,x) \wedge$$ $$ \forall i (0\leq i<k) (\forall x   (\beta^{\pm}(n,m,i,X) \rightarrow B[ X/x])  \wedge  $$  
$$  \forall x,y (  \beta^{\pm}(n,m,i,x) \wedge  \beta^{\pm}(n,m,i+1,y) \rightarrow ( wp(C, X=y) \wedge$$ $$ \neg wp ( C, \bot)  )[X/x]    )      )  ) \rightarrow (  \beta^{\pm}(n,m,k,X) \rightarrow (B \vee  \phi)[X/x]  ) ) $$

\noindent
For more information, readers who are interested may refer to \cite{Winskel93}. Here we stick to the infinite conjunction for ease of proof.

\begin{proposition}\label{precondition}
It holds that $ \vdash \{   wp(C, \phi) \}   C\{\phi\}$.
\end{proposition}

\begin{proof}
We do structural induction on $C$. We only show these non-trivial cases. The cases of (SKIP),(AS) and (PAS) are trivial.

\begin{itemize}

\item (SEQ)  We have $\vdash \{ wp(C_1, wp(C_2, \phi))  \} C_1 \{   wp(C_2, \phi)  \} $ and $ \vdash  \{  wp(C_2, \phi) \} C_2 \{\phi \} $ by inductive hypothesis. Then applying the SEQ rule we know $\vdash \{ wp(C_1, wp(C_2, \phi))  \} C_1 ;C_2 \{    \phi   \} $.

\item (IF) By inductive hypothesis we have $ \{  wp(C_1, \phi) \}C_1 \{\phi \} $ and $ \{  wp(C_2, \phi) \}C_2\{\phi \} $. Then by  (CONS) we know $ \{ B \wedge ( wp(C_1, \phi) \wedge B ) \vee (wp(C_2, \phi) \wedge  \neg B )) \}C_1 \{\phi \} $ and $ \{ \neg  B \wedge ( wp(C_1, \phi) \wedge B ) \vee (wp(C_2, \phi) \wedge  \neg B ))) \}C_2\{\phi \} $. Now by (IF) we have $\vdash  \{( wp(C_1, \phi) \wedge B ) \vee (wp(C_2, \phi) \wedge  \neg B )  \}    if\ B\ then\ C_{1}\ else\ C_{2} \{\phi \} $.

\item (WHILE) It's easy to see that $\models ( B\wedge  \bigwedge\limits_{k \geq 0}  \psi_k  ) \rightarrow (B \wedge  \bigwedge\limits_{k \geq 1}  wp(C, \psi_i)  )   $. By inductive hypothesis, $\vdash  \{\bigwedge\limits_{k \geq 1}  wp(C, \psi_i)\} C \{  \bigwedge\limits_{k \geq 1}    \psi_i \}$. Now by (CONS) we have $\vdash  \{B \wedge  \bigwedge\limits_{k \geq 1}  wp(C, \psi_i)\} C \{  \bigwedge\limits_{k \geq 1}    \psi_i \}$. Note that $\models \bigwedge\limits_{k \geq 1}    \psi_i  \leftrightarrow \bigwedge\limits_{k \geq 0}    \psi_i $. Then by (CONS) we have $\vdash  \{B \wedge  \bigwedge\limits_{k \geq 1}  wp(C, \psi_i)\} C \{  \bigwedge\limits_{k \geq 0}    \psi_i \}$.\\
 That is, $\vdash  \{B \wedge     wp( while \mbox{ }B \mbox{ } do \mbox{ }C  , \phi)  \} C \{  wp( while \mbox{ }B \mbox{ } do \mbox{ }C  , \phi)   \}$.\\
 
 Then, by (WHILE), we know   \\
 
 $\vdash  \{   wp( while \mbox{ }B \mbox{ } do \mbox{ }C  , \phi)  \} while \mbox{ }B \mbox{ } do \mbox{ }C \{  wp( while \mbox{ }B \mbox{ } do \mbox{ }C  , \phi)  \wedge \neg B  \}$.\\
 
 Then by the definition of $wp( while \mbox{ }B \mbox{ } do \mbox{ }C  , \phi)$ we know \\
  $\vdash  \{   wp( while \mbox{ }B \mbox{ } do \mbox{ }C  , \phi)  \} while \mbox{ }B \mbox{ } do \mbox{ }C \{  \phi  \}$.
 
\end{itemize}
\end{proof}

 By soundness and Proposition~\ref{precondition}, we can infer that $\models \{   wp(C, \phi) \}   C\{\phi\}$.  It shows that Definition~\ref{def.weakestprecondition} indeed constructs the precondition given arbitrary postcondition and command. The following proposition states that these $wp(C,\phi)$s are the weakest.

\begin{proposition}\label{weakest precondition}
If $ \models \{\phi\}C\{\psi\}$,  then $\models \phi \rightarrow wp(C, \psi)$.
\end{proposition}

\begin{proof}

\begin{enumerate}
\item Trivial.

\item Assume $ \models\{\phi\}   X \leftarrow E     \{\psi\}$. 

Let $I$ be an arbitrary interpretation. Suppose $\not\models^I  \phi \rightarrow wp(  X \leftarrow E, \psi )$. Then there is a state $S$ such that $S \models^I \phi$ and $S \not\models^I  wp( X \leftarrow E, \psi)$. Then we know $S \not\models^I  \psi[X/ E ]$, which means that $\psi $ is not true under the valuation which assigns $X$ to the value $[\![E]\!]S$ and all other variables to the same value as $S$.  Let $S' = [\![X\leftarrow E]\!](S)=  S[X\mapsto[\![E]\!]S]$. Then $S'$ assigns $X$ to the value $[\![E]\!]S$ and all other variables to the same value as $S$. Then we know $S' \not\models^I \psi$.

By  $\models\{\phi\}   X \leftarrow E     \{\psi\}$ and  $S \models^I \phi$  we know that $[\![X\leftarrow E]\!](S) \models^I \psi$. That is, $S' \models^I \psi$. Contradiction!
 
\item Assume $ \models\{\phi\}   X\xleftarrow{\$}\{a_1: k_{1},...,a_n: k_{n} \}     \{\psi\}$. 
 
Let $I$ be an arbitrary interpretation. Suppose $\not\models^I  \phi \rightarrow wp(  X\xleftarrow{\$}\{a_1: k_{1},...,a_n: k_{n} \} , \psi )$.  Then there is a state $S$ such that $S \models^I \phi$ and $S \not\models^I  wp( X\xleftarrow{\$}\{a_1: k_{1},...,a_n: k_{n} \}  , \psi)$. Then by the construction of weakest precondition we have  $S \not\models^I \psi[X/r_1] \wedge \ldots \wedge  \psi[X/r_n]$.

By  $\models\{\phi\}  X\xleftarrow{\$}\{a_1: k_{1},...,a_n: k_{n} \}     \{\psi\}$ and  $S \models^I \phi$  we know that $[\![  X\xleftarrow{\$}\{a_1: k_{1},...,a_n: k_{n} \} ]\!](S) \models^I \psi$.  Let $\mu' = [\![   X\xleftarrow{\$}\{a_1: k_{1},...,a_n: k_{n}\}     ]\!](\mu_S)$. Then $\mu'$ is a distribution with support $\{S_1, \ldots, S_n\}$, where $S_i =  S[X\mapsto k_i]$ for $i\in \{1, \ldots,n\}$. Hence we know $S_i \models^I \psi $ for all $i \in \{1,\ldots,n\}$.

Note that $S_i$ is a valuation which assigns   $X$ to the value $k_i$ and all other variables to the same value as $S$. Therefore, if $S_i \models^I \psi$ then $S \models^I \psi[X/k_i]$. Then we know $S \models^I \psi[X/k_i]$ for all  $i\in \{1, \ldots,n\}$, which means that $S   \models^I \psi[X/k_1] \wedge \ldots \wedge  \psi[X/k_n]$. Contradiction!

\item Assume $ \models\{\phi\}C_1;C_2\{\psi\}$.
 
Let $I$ be an arbitrary interpretation.  Let $S$ be an arbitrary state such that $S \models^I \phi$. Then we know that $[\![  C_1 ;C_2 ]\!](S) \models^I \psi$. Then  $  [\![  C_2   ]\!]  ( [\![  C_1   ]\!] (S) ) \models^I \psi$. That is, for all $S' \in sp( [\![  C_1   ]\!] (S) )$, $[\![  C_2   ]\!] (S') \models^I \psi$. Let $\phi'$ be the formula which characterizes $S'$, then    $ \models \{\phi' \} C_2 \{ \psi\}$. By induction hypothesis we know $\models \phi' \rightarrow wp( C_2, \psi ) $. Now from $S' \models^I \phi'$ we have $S' \models^I  wp( C_2, \psi )$. From $S' \in sp( [\![  C_1   ]\!] (S) )$ we have $[\![  C_1   ]\!] (S) \models^I wp( C_2, \psi ) $. This proves $\models \{ \phi\} C_1 \{ wp( C_2, \psi ) \}$. Then by  induction hypothesis we know $\models \phi \rightarrow  wp(C_1, wp(C_2, \psi))$.

\item Assume $ \models\{\phi\}  if\ B\ then\ C_{1}\ else\ C_{2}    \{\psi\}$. 

Let $I$ be an arbitrary interpretation.  Let $S$ be an arbitrary state such that $S \models^I \phi$. Then $[\![  if\ B\ then\ C_{1}\ else\ C_{2}   ]\!] (S) \models^I \psi $. It is trivial that either $S\models B$ or $S\models \neg B$. 

\begin{itemize}
\item  
If $S\models B$, then $[\![  if\ B\ then\ C_{1}\ else\ C_{2}   ]\!] (S) = [\![   C_{1}   ]\!] (S) $. Therefore,  $[\![   C_{1}   ]\!] (S) \models^I \psi$.  Let  $\phi'$ be the formula which characterizes $S$, then $S \models^I \phi'$ and  $\models \{\phi'\} C_1 \{\psi \}$. Then by induction hypothesis we know $ \models \phi'  \rightarrow wp( C_1, \psi)$. Therefore, $S \models^I wp( C_1, \psi)$.
\item  
If $S\models \neg B$, then $[\![  if\ B\ then\ C_{1}\ else\ C_{2}   ]\!] (S) = [\![   C_{2}   ]\!] (S) $. Therefore,  $[\![   C_{2}   ]\!] (S) \models^I \psi$.  Let  $\phi'$ be the formula which characterizes $S$, then $S \models^I \phi'$ and  $\models \{\phi'\} C_2 \{\psi \}$. Then by induction hypothesis we know $ \models \phi'  \rightarrow wp( C_2, \psi)$. Therefore, $S \models^I wp( C_2, \psi)$.
\end{itemize}

Sum up the above two cases we have $S \models^I (B\wedge   wp( C_1, \psi) ) \vee (\neg B\wedge   wp( C_2, \psi) )$. This proves $\models \phi \rightarrow  ( (B\wedge   wp( C_1, \psi) ) \vee (\neg B\wedge   wp( C_2, \psi) ) )$.

\item Assume $ \models \{\chi\}    while \mbox{ }B \mbox{ } do \mbox{ }C      \{\phi\}$ and the command will terminate if it starts from all states where $\chi$ is satisfiable. Let $I$ be an interpretation and $S$ be a state such that $S\models^I \chi$. We prove by induction on the times of the iteration of the while loop.

\begin{itemize}

\item 
Suppose the while loop is terminated after  $0$ time of execution.\footnote{The loop terminated after $k$ steps means that the longest branch of the execution of the while loop terminates after $k$ steps. } Then  $   S\not\models B   $ and       $   [\![         while \mbox{ }B \mbox{ } do \mbox{ }C      ]\!]   (\mu_S )  = \mu_S $. Therefore, by  $ \models \{\chi\}    while \mbox{ }B \mbox{ } do \mbox{ }C      \{\phi\}$ we know  $S \models^I \phi$. Then we have $S\models^I \neg B \wedge \phi$. Hence $S \models^I \psi_i$ for all $i$. Therefore, $S \models^I wp( while \mbox{ }B \mbox{ } do \mbox{ }C  , \phi)$.

\item  
Suppose the while loop is terminated after  $1$ time of execution. Then  $   S\models B   $ and $ ([\![   C  ]\!] \circ \downarrow_{ B})(\mu_S ) = [\![   C  ]\!] (\mu_S )$. Therefore, for all $S' \in sp ( [\![   C  ]\!] (\mu_S ) )$,  $   S' \not\models B$, $S'  \models^I  \phi$ and $[\![   C  ]\!] (\mu_S )  \models^I \neg B\wedge \phi$. 
Hence $[\![   C  ]\!] (\mu_S )  \models^I \psi_i$ for all $i$.
Let $\chi'$ be the formula that characterizes $S$. Then we have $\models \{\chi'\} C \{\psi_i\}$. 
Then by induction hypothesis we know $ \models \chi' \rightarrow wp(C, \psi_i)$. Then from   $S\models^I \chi'$ and $S\models B$ we know
 $S \models^I \psi_i$ for all $i$. Therefore, $S \models^I wp( while \mbox{ }B \mbox{ } do \mbox{ }C  , \phi)$.

\item  
Suppose the while loop is terminated after  $k+1$ times of execution. 

Let $S_1 \in sp (   ([\![   C   ]\!] \circ \downarrow_{ B})^k (\mu_S )  )$. 

\begin{itemize}

\item If $S_1 \models \neg B$, then $S_1$ is a state where the loop terminates after $k$ time of execution. Then by the induction hypothesis we know $S_1\models^I \psi_i$ for all $i$.

\item If $S_1 \models  B$, then $ ([\![   C  ]\!] \circ \downarrow_{ B})(\mu_{S_1} ) = [\![   C  ]\!] (\mu_{S_1} )$. Therefore, for all $S' \in sp ( [\![   C  ]\!] (\mu_{S_1} ) )$,  $   S' \not\models B$, $S'  \models^I  \phi$ and $[\![   C  ]\!] (\mu_{S_1} )  \models^I \neg B\wedge \phi$. 
Hence $[\![   C  ]\!] (\mu_{S_1} )  \models^I \psi_i$ for all $i$.
Let $\chi'$ be the formula that characterizes ${S_1}$. Then we have $\models \{\chi'\} C \{\psi_i\}$. 
Then by induction hypothesis we know $ \models \chi' \rightarrow wp(C, \psi_i)$. Then from   ${S_1}\models^I \chi'$ and $S_1\models B$ we know
 ${S_1} \models^I \psi_i$ for all $i$.
\end{itemize}

This proves $ ([\![   C   ]\!] \circ \downarrow_{ B})^k (\mu_S )   \models^I \psi_i $ for all $i$.  

Let $S_2 \in sp(  ([\![   C   ]\!] \circ \downarrow_{ B})^{k-1} (\mu_S )    )$.

Now, we split into two cases: 

\begin{itemize}
\item if $S_2 \models \neg B$, then $S_2\models^I \neg B \wedge \phi$.
\item 
If $S_2 \models B$, then  $([\![   C   ]\!] \circ \downarrow_{ B}) (\mu_{S_2} ) = [\![   C   ]\!]  (\mu_{S_2} )$. Then $ [\![   C   ]\!]  (\mu_{S_2} )   \models^I \psi_i $ for all $i$. 
Let $\chi'$ be the formula that characterizes ${S_2}$. Then we have $\models \{\chi'\} C \{\psi_i\}$. 
Then by induction hypothesis we know $ \models \chi' \rightarrow wp(C, \psi_i)$. Then from   ${S_2}\models^I \chi'$ and $S_2 \models B$ we know
 ${S_2} \models^I \psi_i$ for all $i$.
\end{itemize}

Hence $S_2 \models^I wp(C, \psi_i)$.

This proves $ ([\![   C   ]\!] \circ \downarrow_{ B})^{k-1} (\mu_S    )   \models^I \psi_i $ for all $i$.  

Repeat the above procedure we deduce that  $S \models^I B \wedge wp(C, \psi_i) $ for all $i$.

\item Now we study the case in which  the loop never terminates. That is, there is a infinite sequence of deterministic states $S=S_0, S_1, S_2,\ldots$ in which $S_{j+1} \in sp ([\![ C ]\!](S_j) )$. We will show that $S_j \models \psi_i$ for all $i$ and $j$.

It's easy to see that $S_j \models B$ for all $j$ because otherwise the loop will terminate. Now we  prove that for all $j$, $S_j \models \psi_i$ for all $i$ by induction on $i$.

It's easy to see that $S_j \models \psi_0$ for all $j$. Assume for all $j$ it holds that $S_j \models \psi_k$. Then $S_{j-1} \models wp(C, \psi_k)$. Therefore, $S_{j-1} \models \psi_{k+1}$.

\end{itemize}   
\end{enumerate}

\end{proof}





\begin{corollary}\label{corollary.weakestprecondition}
 $S \models wp(C, \phi)$ iff $[\![  C   ]\!](S) \models \phi$.
\end{corollary}

\begin{proof} \ \\
\begin{itemize}
\item ($\Rightarrow $)  If $S \models wp(C, \phi)$, then by Proposition~\ref{precondition} and soundness,  we have $\models \{   wp(C, \phi) \}   C\{\phi\}$ which implies that  $[\![  C   ]\!](S) \models \phi$.

\item ($\Leftarrow $) 
  Assume $[\![  C   ]\!](S) \models \phi$. Let $\phi'$ be the formula which characterizes $S$. Then we know  $S \models \phi'$  and $\models  \{  \phi' \} C  \{  \phi   \}$.  Then  we have $\models \phi' \rightarrow wp(C, \psi)$. Hence $S  \models wp(C, \phi)$.
\end{itemize} 
\end{proof}

After combining Proposition~\ref{weakest precondition} and Corollary~\ref{corollary.weakestprecondition}, it is concluded that $wp(C,\phi)$ defined in Definition~\ref{def.weakestprecondition} is the weakest precondition for command $C$ and the postcondition $\phi$. This leads to an easy completeness proof.

\begin{corollary}[Completeness]\label{th.completeness}

If $\models \{\phi \} C \{ \psi\}$, then $\vdash \{\phi \} C \{ \psi\}$.

\end{corollary}
 
\begin{proof} 
If $\models \{\phi \} C \{ \psi\}$, then by Proposition \ref{weakest precondition} we know $\models \phi \rightarrow  wp(C, \psi)$. By Proposition \ref{precondition} we know $\vdash \{ wp(C, \psi) \} C \{\psi \}$. Then apply the derivation rule (CONS) we have  $\vdash \{\phi \} C \{ \psi\}$.   
\end{proof}

So far, we have demonstrated that PHL$_d$ is both sound and complete. However, Hoare triples only deal with deterministic formulas. Probabilistic assertions, on the other hand, are more expressive because they can describe the probabilistic nature of probabilistic states, {\em e.g.}, the probability of $X>3$ is $\frac{1}{2}$. In the next section, we are going to investigate probabilistic assertions and their corresponding proof system. Before we start, it is worth noting that the relationship between WHILE command and IF command shown in the lemma below, which plays an important role in the completeness proof in the next section.

\begin{lemma}\label{key lemma}

For all $i\geq 0$, if $S \models \neg wp(C^0 , \neg B ) \wedge \ldots \wedge    \neg wp(C^{i-1} , \neg B ) \wedge   wp(C^{i} , \neg B ) $, then $ [\![  \texttt{while} \mbox{ } B \mbox{ } \texttt{do} \mbox{ } C  ]\!] \mu_S =  [\![  ( \texttt{if}\mbox{ } B \mbox{ } \texttt{then}\mbox{ }  C \mbox{ }  \texttt{else}\mbox{ }  \texttt{skip} )^{i} ]\!] \mu_S$.

\end{lemma}

\begin{proof}

We first prove the case where $i=0$. Here we have $S \models wp(C^{0} , \neg B)$. That is, $S \models \neg B$. It is then trivial to prove $ [\![  \texttt{while} \mbox{ } B \mbox{ } \texttt{do} \mbox{ } C  ]\!] \mu_S =  [\![  ( \texttt{if}\mbox{ } B \mbox{ } \texttt{then}\mbox{ }  C \mbox{ }  \texttt{else}\mbox{ }  \texttt{skip} )^{0} ]\!] \mu_S$.

We now prove the case with $i=2$, other cases are similar.
Assume $i=2$. Then we have $S\models \neg wp(C^0,\neg B) \wedge \neg wp(C^1, \neg B) \wedge wp(C^2 ,\neg B)$. 	Then we know $S \not\models wp( \texttt{skip} , \neg B)$, $S \not\models wp(  C   , \neg B)$  and $S\models wp(C;C ,\neg B)$. Therefore, $S\models \neg B$, $[\![C ]\!] S \not \models  \neg B$ and $[\![C;C ]\!] S \models  \neg B$. We then know $\downarrow_{  B} ( [\![C;C ]\!] S)  =  \textbf{0}  $.

It's easy to see that $sp(     [\![   C   ]\!] \circ \downarrow_{ B}  (\mu_S ) ) \subseteq sp(    [\![   C   ]\!]    (\mu_S ) ) $. Then from $S \models wp(C^2 , \neg B)$ we deduce  $[\![   C^2   ]\!] S \models \neg B $, which implies $[\![   (C \circ  \downarrow_{ B})^2   ]\!] S \models \neg B $.

Then $\displaystyle\sum_{i = 0}^{\infty} \downarrow_{\neg B}  (  (  [\![   C   ]\!] \circ \downarrow_{ B})^i (\mu_S ) )  =   \downarrow_{\neg B}  (  (  [\![   C   ]\!] \circ \downarrow_{ B})^0 (\mu_S ) )  + \downarrow_{\neg B}  (  (  [\![   C   ]\!] \circ \downarrow_{ B})^1 (\mu_S ) )  + \downarrow_{\neg B}  (  (  [\![   C   ]\!] \circ \downarrow_{ B})^2 (\mu_S ) ) + \downarrow_{\neg B}  (  (  [\![   C   ]\!] \circ \downarrow_{ B})^3 (\mu_S ) ) + \ldots  =  \downarrow_{\neg B}  (  (  [\![   C   ]\!] \circ \downarrow_{ B})^0 (\mu_S ) )  + \downarrow_{\neg B}  (  (  [\![   C   ]\!] \circ \downarrow_{ B})^1 (\mu_S ) )  + \downarrow_{\neg B}  (  (  [\![   C   ]\!] \circ \downarrow_{ B})^2 (\mu_S ) )  =  \downarrow_{ \neg B} (\mu_S  ) +  \downarrow_{ \neg B}\circ  [\![   C   ]\!] \circ \downarrow_{ B} ( \mu_S  )  + ([\![   C   ]\!] \circ \downarrow_{ B} )^2  (\mu_S  )  $.

On the other hand,  $[\![\texttt{if}  \mbox{ }  B   \mbox{ }  \texttt{then}   \mbox{ }  C  \mbox{ }  \texttt{else}  \mbox{ } \texttt{skip}  ]\!]^{2}   \mu_S  =  $\\
$ [\![\texttt{if}  \mbox{ }  B   \mbox{ }  \texttt{then}   \mbox{ }  C  \mbox{ }  \texttt{else}  \mbox{ } \texttt{skip}  ]\!]   (   [\![   C   ]\!] \circ \downarrow_{ B} (\mu_S )  + \downarrow_{ \neg B} (\mu_S ) )  =  $\\
$[\![   C   ]\!] \circ \downarrow_{ B} ([\![   C   ]\!] \circ \downarrow_{ B} (\mu_S )  + \downarrow_{ \neg B} (\mu_S ) )  + \downarrow_{ \neg B} ([\![   C   ]\!] \circ \downarrow_{ B} (\mu_S )  + \downarrow_{ \neg B} (\mu_S ) ) = $\\
$[\![   C   ]\!] \circ \downarrow_{ B}  [\![   C   ]\!] \circ \downarrow_{ B} (\mu_S )  +  [\![   C   ]\!] \circ \downarrow_{ B} (\downarrow_{ \neg B} (\mu_S ) )  + \downarrow_{ \neg B} [\![   C   ]\!] \circ \downarrow_{ B} (\mu_S )  + \downarrow_{ \neg B} (\downarrow_{ \neg B} (\mu_S ) )  = $\\
$ ([\![   C   ]\!] \circ \downarrow_{ B} )^2  (\mu_S )  +  \textbf{0}  + \downarrow_{ \neg B} [\![   C   ]\!] \circ \downarrow_{ B} (\mu_S )  +   \downarrow_{ \neg B} (\mu_S )  =$\\
$   \downarrow_{ \neg B} (\mu_S )  +  \downarrow_{ \neg B} [\![   C   ]\!] \circ \downarrow_{ B} (\mu_S )  + ([\![   C   ]\!] \circ \downarrow_{ B} )^2  (\mu_S )  $.

\end{proof}

\section{Probabilistic Hoare Logic with probabilistic formulas}\label{Probabilistic Hoare Logic with probabilistic assertion}

\subsection{Probabilistic formulas}

In order to describe the probabilistic aspects of probabilistic states, we need to `expand' deterministic formulas to probabilistic formulas. In this section, we first define real expressions (a.k.a. probabilistic expression) which are building blocks of probabilistic formulas.

\begin{definition}[Real expressions] Let $\mathbb{RV}$ be a set of real variables. The real expression $r$ is defined as follows:

$$r:=  a \mid      \mathfrak{x}      \mid  \mathbb{P}(\phi ) \mid r\mbox{ } aop \mbox{ } r $$
 
\end{definition}

Here $a \in \mathbb{R}$ is a real number and $\mathfrak{x}   \in \mathbb{RV}$ is a real variable.  $\mathbb{P}(\phi )$ is the probability of a \textit{deterministic assertion} $\phi$ being true. Given an interpretation $I$ which maps logical variables to integers and maps real variables to real numbers, the semantics of real expressions are defined in the standard way as follows.

\begin{definition}[Semantics of real expressions]  Given an interpretation $I$ and a probabilistic state $\mu$, the semantics of real expressions is defined inductively as follows.

\begin{itemize}
\item $[\![a ]\!]^I_{\mu} = a $.

 \item $[\![  \mathfrak{x}     ]\!]^I_{\mu} = I(  \mathfrak{x}   ) $.

\item $[\![ \mathbb{P}(\phi )  ]\!]^I_{\mu}  = \sum\limits_{S\models^I \phi} \mu (S) $.

\item $[\![ r_1\mbox{ } aop \mbox{ } r_2   ]\!]^I_{\mu} =   [\![r_1  ]\!]^I_{\mu}    \mbox{ } aop \mbox{ }    [\![r_2  ]\!]^I_{\mu}  $.
\end{itemize}
\end{definition}

A real number $a$ is interpreted as itself over arbitrary probabilistic state. $I$ decides the value of each real variable $\mathfrak{x}$. The expression $\mathbb{P}(\phi)$ represents the probability that the deterministic formula $\phi$ holds over some probabilistic state. It can be computed by summing up all probabilities of these deterministic states that make $\phi$ true. $r_1\ aop\ r_2$ characterizes the arithmetic calculation between two real expressions. 

Now we define probabilistic formulas (assertions). The primitive type of probabilistic formula is to show the relationship between two real expressions. They are often used for comparing two probabilities, e.g. $\mathbb{P}(X>1)<\frac{1}{2}$. 

\begin{definition}[Syntax of probabilistic formulas]  Probabilistic formulas are defined inductively as follows.

$$\Phi=    (r \mbox{ }  rop \mbox{ }   r) \mid \neg \Phi \mid (\Phi \wedge \Phi)$$

where $r$ is a real expression. 
\end{definition}

For example,  $(\mathbb{P}(X>0)>\frac{1}{2})\wedge \neg (\mathbb{P}(X>1)<\mathbb{P}(X>0))$ is a well-formed probabilistic formula. 

\begin{definition}[Semantics of probabilistic formulas]\label{def.semanticsprobabilisticformulas}
Given an interpretation $I$, the semantics of probabilistic assertion is defined on probabilistic states $\mu$ as follows:  
 
\begin{itemize}
\item $\mu \models^I   r_1  \mbox{ }  rop \mbox{ }     r_2 $ if  $  [\![r_1]\!]^I_{\mu} \mbox{ }  rop \mbox{ }   [\![r_2]\!]^I_{\mu} = \top$ 

\item $\mu \models^I \neg \Phi$ if not $\mu \models^I  \Phi$

\item $\mu \models^I  \Phi_1 \wedge \Phi_2$ if   $\mu \models^I  \Phi_1$ and $\mu \models^I  \Phi_2$

\end{itemize} 
\end{definition}

\subsection{Proof system with probabilistic formulas}

As mentioned in Section 2, the proof system of probabilistic Hoare logic consists of Hoare triples. Now we step forward to Hoare triples $\{\Phi\}C\{\Psi\}$ where $\Phi$ and $\Psi$ are probabilistic formulas. The main technical difficulty is also finding out weakest preconditions for given command and probabilistic formulas. Inspired by Chadha \textit{et al.} \cite{chadha2007reasoning}, we first introduce the notion {\em weakest preterms} which constitute weakest preconditions. 

The terminology `term' just refers to the real expression in this paper. We inherit this word from Chadha \textit{et al.}. A weakest preterm with respect to a given command $C$ and a term (real expression) $r$ intuitively denotes the term whose interpretation of the initial probabilistic state is the same as the interpretation of $r$ on the resulting state after executing $C$. We also borrow the notion of the conditional term from Chadha \textit{et al.} \cite{chadha2007reasoning}, which will be useful in defining the weakest preterm.

\begin{definition}[Conditional terms] The conditional term $r/B$ of a real expression  $r$ with a Boolean formula $B$ (a deterministic formula) is inductively defined as follows.
 \begin{itemize}
 
\item $a/B =a $

 \item $ \mathfrak{x}/B = \mathfrak{x} $ 
 
 \item $  \mathbb{P}(\phi ) /B =  \mathbb{P}(\phi \wedge B)$
 
 \item $(r_1\mbox{ } aop \mbox{ } r_2 ) /B = r_1/B \mbox{ } aop \mbox{ } r_2/B$
 
 \end{itemize}
\end{definition}

A conditional term intuitively denotes a probability under some condition described by a deterministic formula, which is shown by the next lemma.

 \begin{lemma}
 
 $ [\![ r /B ]\!]^I_{\mu} = [\![ r  ]\!]^I_{\downarrow_{B}  \mu}  $

 \end{lemma}
 
 \begin{proof}

The only non-trivial case is when the term $r$ of the form $ \mathbb{P}(\phi ) $. In this case, we have

$ [\![  \mathbb{P}(\phi )  /B ]\!]^I_{\mu} = [\![ \mathbb{P}(\phi \wedge B) ]\!]^I_{\mu}=   [\![  \mathbb{P}(\phi )   ]\!]^I_{\downarrow_{B}  \mu}  $

\end{proof}

Now, we are fully prepared to define the weakest preterms. The weakest preterms for real expression $a$, $\mathfrak{x}$ and $r_1\ aop\ r_2$ are straightforward. In terms of $\mathbb{P}(\phi)$, we need to split cases by different commands among which WHILE command brings a lot of subtleties.

\begin{definition}[weakest preterms]\label{def.weakestpreterm}
The weakest preterm of a real expression  $r$ with command $C$ is inductively defined as follows.

\begin{enumerate}

\item $pt(C,a)=a$

\item $pt(C,  \mathfrak{x}   )=  \mathfrak{x}   $

\item $pt(C,r_1\mbox{ } aop \mbox{ } r_2 )= pt(C, r_1) \mbox{ } aop \mbox{ } pt(C, r_2)$

\item  $ pt (\texttt{skip}, \mathbb{P}(\phi  )) = \mathbb{P}(  \phi )$

\item  $ pt (X\leftarrow E , \mathbb{P}(\phi  )) = \mathbb{P}( \phi[X/E] )$ 

\item $ pt (   X\xleftarrow[]{\$}  \{a_1: k_1, \ldots ,a_n:k_n\} ,       \mathbb{P}(\phi  )) =$\\ 
$ a_1    \mathbb{P} ( \phi[X/k_1] \wedge \neg  \phi[X/k_2] \wedge \ldots \neg \phi[X/k_n]    )  + $\\ 
$a_2    \mathbb{P} ( \neg \phi[X/k_1] \wedge    \phi[X/k_2] \wedge \neg \phi[X/k_3] \wedge  \ldots  \wedge \neg \phi[X/k_n]    ) +  \ldots + $
$a_n    \mathbb{P} ( \neg \phi[X/k_1] \wedge     \ldots \wedge \neg \phi[X/k_{n-1}] \wedge    \phi[X/k_n]  ) + $\\ 
$(a_1 +a_2) \mathbb{P} ( \phi[X/k_1] \wedge   \phi[X/k_2] \wedge  \neg   \phi[X/k_3] \wedge \ldots \wedge \neg \phi[X/k_n]    ) + \ldots +  (a_1+ \ldots + a_n)  \mathbb{P} ( \phi[X/k_1] \wedge    \ldots \wedge \phi[X/k_n]    )  $

\item $pt(C_1;C_2 ,\mathbb{P}(\phi  ) ) = pt (C_1, pt(C_2, \mathbb{P}(\phi  ))) $

\item  $ pt (\texttt{if}\ B\ \texttt{then}\ C_{1}\ \texttt{else}\ C_{2} ,  \mathbb{P}(\phi  )  ) =  pt(C_1,  \mathbb{P}(\phi  )  )/B   + \\ pt(C_2,  \mathbb{P}(\phi  )  )/(\neg B)  $  

\item $ pt( \texttt{while} \mbox{ } B \mbox{ } \texttt{do} \mbox{ } C ,   \mathbb{P}(\phi  ))   =  \displaystyle\sum_{i = 0}^{\infty}  T_i $, in which $T_i$ is defined via the following procedure.  We  use the following abbreviation

\begin{itemize}

\item $wp(i) $ is short for $  \neg wp(C^0 , \neg B ) \wedge \ldots \wedge    \neg wp(C^{i-1} , \neg B ) \wedge   wp(C^{i} , \neg B ) $

\item $wp(\infty) $ is short for $ \displaystyle \bigwedge_{i=0}^{ \infty}  \neg  wp(C^{i} , \neg B ) $

\item $WL$ is short for $\texttt{while} \mbox{ } B \mbox{ } \texttt{do} \mbox{ } C $

\item $IF$ is short for $\texttt{if}\ B\ \texttt{then}\ C \ \texttt{else} \mbox{ }  \texttt{skip}$

\item   $SUM$ is short for $\displaystyle\sum_{i = 0}^{\infty} (  \mathbb{P}(    wp(i)    )      ( pt ( (  IF)^i  ,\mathbb{P}(\phi  )  ) / (    wp(i)   )) )$

\end{itemize}

For any probabilistic state $\mu$, we let\\
 $\mu_0 = \downarrow_{wp(\infty)}(\mu)$, $\mu_{i+1}=\downarrow_{wp(\infty)}(  [\![C]\!] \mu_i) $.
 
 \begin{itemize}
\item   Let $T_0 = SUM$.
\item Let $T_1$ be the unique real expression  such that 
 $$[\![  T_1]\!]_{\mu } = [\![   \mathbb{P}(  wp(\infty)  )     ]\!]_{\mu }  [\![ SUM]\!]_{ [\![ C]\!] \mu_0 }      $$ for all probabilistic state $\mu$.
Equivalently,  $$T_1  =   \mathbb{P}(  wp(\infty)  )     (       pt(C,SUM) /wp(\infty)        )   .$$
\item Let $T_2$   be the unique real expression  such that 
$$[\![  T_2]\!]_{\mu } = [\![   \mathbb{P}(  wp(\infty)  )     ]\!]_{\mu }     \mathbb{P}(  wp(\infty)  )     ]\!]_{ [\![ C]\!] \mu_0 }          [\![ SUM]\!]_{ [\![ C]\!] \mu_1 }      $$
 for all probabilistic state $\mu$. Equivalently, $T_2=$
 $$  \mathbb{P}(  wp(\infty)  )         (       pt(C,      \mathbb{P}(  wp(\infty)  )    ) /wp(\infty)        ) (            pt(C,             pt(C,SUM) /wp(\infty)          ) /wp(\infty)                         )   $$
 
 \item  In general, we let $T_i$   be the unique real expression  such that $[\![  T_i]\!]_{\mu } =$
 $$  [\![   \mathbb{P}(  wp(\infty)  )     ]\!]_{\mu }      [\![\mathbb{P}(  wp(\infty)  )     ]\!]_{ [\![ C]\!] \mu_0 }   \ldots $$
 $$  [\![\mathbb{P}(  wp(\infty)  )     ]\!]_{ [\![ C]\!] \mu_{i-2} }  [\![ SUM]\!]_{ [\![ C]\!] \mu_{i-1} }  $$
 for all probabilistic states $\mu$. Equivalently, let $f_{C,B}$ be the function which maps a real expression $r$ to $f_{C,B}(r) = pt(C,r) / wp(\infty)  $. Then 
 
 $$T_i = f_{C,B}^i(SUM)  \displaystyle\prod_{j = 0}^{i-1}     f_{C,B}^j ( \mathbb{P}(  wp(\infty)  ) ).  $$
 
 Therefore, $$ pt( \texttt{while} \mbox{ } B \mbox{ } \texttt{do} \mbox{ } C ,   \mathbb{P}(\phi  ))   =  \displaystyle\sum_{i = 0}^{\infty}   f^i(SUM)  \displaystyle\prod_{j = 0}^{i-1}     f^j ( \mathbb{P}(  wp(\infty)  ) ).  $$
 
 \end{itemize}

\end{enumerate}

\end{definition}

Let us explain the definition of $pt(\texttt{while} \mbox{ } B \mbox{ } \texttt{do} \mbox{ } C ,   \mathbb{P}(\phi  ))$. Given a deterministic state $S$, $S\models wp(i)$ means that $S$ will  lead to a probabilistic state satisfying $\neg B$ (all its supports satisfy $\neg B$) after executing $C$ for \textit{exactly} $i$ times. Thus, $SUM$ represents the probability of reaching a probabilistic state satisfying $\neg B$ by executing $C$ in  finitely many steps. In contrast, $S\models wp(\infty)$ means that $S$ will never achieve a probabilistic state satisfying  $\neg B$. So, given a probabilistic state, some supports of it satisfy $wp(i)$ for some $i$ (the probability of all of them is denoted by $SUM$), but others satisfy $wp(\infty)$. For a deterministic state which satisfies $wp(\infty)$, we execute $C$ for 1 time to get a probabilistic state, of which the support   is a  set of deterministic states. Among those deterministic states, we   use $T_1$ to represent those ones which will terminate on $\neg B$ after finite times of execution. For those states which will not terminate on $\neg B$ in finite times, we execution $C$ for the 2nd time. Then $T_2$ will pop up. So on so forth. In this way we construct $T_i$ for each $i\geq 1$.

Although according to the above definition, $pt(C,r)$ is a syntactic notion. In practice we will always treat it as a semantic notion. That is, if $[\![ pt(C,r_1) ]\!]_{\mu}^{I} = [\![ r_2 ]\!]_{\mu}^{I}  $ for all probabilistic state $\mu$ and interpretation $I$, then we say that $r_2$ is a preterm of  $r_1$ with command $C$. The above definition constitute the weakest preterm calculus of probabilistic expression. It tells us how to calculate the weakest preterm for any give probabilistic expression and command. Now we use some examples to demonstrate the usage of the  weakest preterm calculus.

\begin{example} Let's calculate $ pt( \texttt{while} \mbox{ } \top \mbox{ } \texttt{do} \mbox{ } \texttt{skip} ,   \mathbb{P}( \top  ))$. In this case we have $wp(\texttt{skip}^{i} , \neg \top ) =  wp(\texttt{skip}  ,  \bot ) =\bot $. (Here we overload the symbol $=$ to represent semantic equivalence.) Then we know $wp(i) =\bot$ for all $i$ and $wp(\infty)=\top$. Therefore, $SUM=0$ because $ \mathbb{P}(    wp(i)    ) =  \mathbb{P}(   \bot   ) =0$. Moreover, $T_1  =   \mathbb{P}(  wp(\infty)  )     (       pt(  \texttt{skip} , $  $SUM) /wp(\infty)        )   =   \mathbb{P}(  \top  )     (       pt(  \texttt{skip} , 0 ) / \top       )  =  \mathbb{P}(  \top  )     \cdot 0  =0  $. We further note that $f_{ \texttt{skip}  , \top }(r) = pt( \texttt{skip} ,r) /\top =r $. Therefore, $T_i =    0 $ for all $i$. We then know $ pt( \texttt{while} \mbox{ } \top \mbox{ } \texttt{do} \mbox{ } \texttt{skip} ,   \mathbb{P}( \top  )) = 0$.

\end{example}

\begin{example} Let $C$ be the following program: 
$$X \leftarrow \{\frac{1}{3} :0, \frac{2}{3} :1\}  ; \texttt{if}\ X = 0 \ \texttt{then}\  (\texttt{while} \mbox{ } \top \mbox{ } \texttt{do} \mbox{ } \texttt{skip})   \ \texttt{else}\ \texttt{skip}   $$
\noindent
Let's calculate $ pt( C ,   \mathbb{P}( \top  )    )$. We have 
$$pt(C,  \mathbb{P}( \top  )) = pt (  X \leftarrow \{\frac{1}{3} :0, \frac{2}{3} :1\}    , $$  $$pt ( \texttt{if}\ X = 0 \ \texttt{then}\  (\texttt{while} \mbox{ } \top \mbox{ } \texttt{do} \mbox{ } \texttt{skip})   \ \texttt{else}\ \texttt{skip}  ,   \mathbb{P}( \top  )    )   ) .$$
$$  pt (  \texttt{if}\ X = 0 \ \texttt{then}\  (\texttt{while} \mbox{ } \top \mbox{ } \texttt{do} \mbox{ } \texttt{skip})   \ \texttt{else}\ \texttt{skip}  ,  \mathbb{P}( \top  ) )= $$
$$(pt (\texttt{while} \mbox{ } \top \mbox{ } \texttt{do} \mbox{ } \texttt{skip},   \mathbb{P}( \top  ) )   /(X=0 ) ) + (pt (  \texttt{skip},   \mathbb{P}( \top  ) )  / (\neg X=0) )=$$
$$( 0 /(X=0 ) ) + ( pt (  \texttt{skip},   \mathbb{P}( \top  ) )  / (\neg X=0))=$$
$$ 0 +    (\mathbb{P}( \top  )   / (\neg X=0))=    \mathbb{P}(   X \neq 0  ). $$
 Therefore, $pt(C,  \mathbb{P}( \top  )) = pt( X \leftarrow \{\frac{1}{3} :0, \frac{2}{3} :1\} ,  \mathbb{P}(  X\neq 0  ) )= $

$\frac{1}{3}  \mathbb{P}( 0\neq 0 \wedge \neg 1 \neq 0         ) + \frac{2}{3}  \mathbb{P}(        \neg  0\neq 0 \wedge 1\neq 0 )  + (\frac{1}{3} +\frac{2}{3}) \mathbb{P}(     0\neq 0 \wedge 1\neq 0     ) = \frac{1}{3}  \mathbb{P}( \bot) +  \frac{2}{3}  \mathbb{P}( \top ) + \mathbb{P}( \bot) =  \frac{2}{3}  \mathbb{P}( \top ) $.

\end{example}

Now we prove the characterization lemma of preterms. This lemma indicates that the preterm of a real expression $r$ (can be understood as a probability) in terms of command $C$ on a given probabilistic state $\mu$ equals $r$ after executing $C$.

\begin{lemma} \label{preterm lemma}

Given a probabilistic state $\mu $, an interpretation $I$, a command $C$ and a real expression $r$,   

   $$ [\![  pt  ( C,  r ) ]\!]^I_{\mu } =  [\![ r ]\!]^I_{ [\![ C]\!] \mu } $$

\end{lemma}

\begin{proof}

\begin{enumerate}

\item $ [\![ pt(C,  a  ) ]\!]^I_{\mu } = [\![  a ]\!]^I_{\mu } =  a =  [\![  a ]\!]^I_{ [\![ C]\!] \mu }$.

\item $ [\![ pt(C,  \mathfrak{x}   ) ]\!]^I_{\mu } = [\![  \mathfrak{x}   ]\!]^I_{\mu } =  I( \mathfrak{x}  ) =  [\![  \mathfrak{x}   ]\!]^I_{ [\![ C]\!] \mu }$.

\item $ [\![ pt(C, r_1\mbox{ } aop \mbox{ } r_2 ) ]\!]^I_{\mu }    [\![   pt(C, r_1) \mbox{ } aop \mbox{ } pt(C, r_2)     ]\!]^I_{\mu } =   $\\   $  [\![   pt(C, r_1) \mbox{ }   ]\!]^I_{\mu } aop [\![    \mbox{ } pt(C, r_2)     ]\!]^I_{\mu }  =        [\![  r_1  ]\!]^I_{ [\![ C]\!] \mu }   aop    [\![  r_2  ]\!]^I_{ [\![ C]\!] \mu }  $.    

\item    $   [\![  pt  ( \texttt{skip},  \mathbb{P} (\phi) ) ]\!]^I_{\mu }  =   [\![ \mathbb{P} (  \phi  )]\!]^I_{\mu } =    [\![ \mathbb{P} (  \phi  )]\!]^I_{ [\![ \texttt{skip}]\!] \mu } $. 

\item   $ [\![  pt  ( X\leftarrow E ,  \mathbb{P} (\phi) ) ]\!]^I_{\mu } = [\![ \mathbb{P} (   \phi[X/E] )]\!]^I_{\mu }    =  [\![ \mathbb{P} (wp(    X\leftarrow E    , \phi) )]\!]^I_{\mu } $. Since $S \models \phi[X/E] $ iff $ S\models  wp(    X\leftarrow E    , \phi) $ iff $  [\![  X\leftarrow E ]\!] S \models  \phi  $, we know $ [\![ \mathbb{P} (   \phi[X/E] )]\!]^I_{\mu } =  [\![ \mathbb{P} (   \phi  )]\!]^I_{  [\![  X\leftarrow E ]\!] \mu }$.

\item PAS: For the sake of simplicity, we assume $n=2$. No generality is lost with this assumption. We have $ [\![  pt  (  X\xleftarrow[]{\$}  \{a_1: k_1,  a_2:k_2\} ,  \mathbb{P} (\phi) ) ]\!]^I_{\mu }  =   [\![  a_1  \mathbb{P} (\phi[X/k_1 ] \wedge \neg \phi[X/k_2 ] ) +  a_2 \mathbb{P} ( \neg \phi[X/k_1 ] \wedge   \phi[X/k_2 ] ) + (a_1 +a_2 )  \mathbb{P} (\phi[X/k_1 ] \wedge  \phi[X/k_2 ] ) ]\!]^I_{\mu } $. Without loss generality, assume $sp(\mu) =\{S_1, S_2, S_3, S_4\}$, $S_1 \models^I \phi[X/k_1 ] \wedge \neg \phi[X/k_2 ] $, $S_2 \models^I \neg \phi[X/k_1 ] \wedge   \phi[X/k_2 ] $, $S_3 \models^I \phi[X/k_1 ] \wedge   \phi[X/k_2 ] $, $S_4 \models^I \neg \phi[X/k_1 ] \wedge \neg \phi[X/k_2 ] $, $\mu(S_i) = b_i$. Then we have $  \![  a_1  \mathbb{P} (\phi[X/k_1 ] \wedge \neg \phi[X/k_2 ] ) +  a_2 \mathbb{P} ( \neg \phi[X/k_1 ] \wedge   \phi[X/k_2 ] ) + (a_1 +a_2 )  \mathbb{P} (\phi[X/k_1 ] \wedge  \phi[X/k_2 ] ) ]\!]^I_{\mu } =  a_1 b_1 + a_2 b_2 + (a_1 +a_2)b_3$.

We further have  $[\![X\xleftarrow[]{\$}  \{a_1: k_1,  a_2:k_2\}]\!] (S_1) \models^I \mathbb{P} (\phi) = a_1 $,  $[\![X\xleftarrow[]{\$}  \{a_1: k_1,  a_2:k_2\}]\!] (S_2) \models^I \mathbb{P} (\phi) = a_2 $,  $[\![X\xleftarrow[]{\$}  \{a_1: k_1,  a_2:k_2\}]\!] (S_3) \models^I \mathbb{P} (\phi) = a_1 +a_2 $ and $[\![X\xleftarrow[]{\$}  \{a_1: k_1,  a_2:k_2\}]\!] (S_4) \models^I \mathbb{P} (\phi) = 0 $. Then we know $[\![   X\xleftarrow[]{\$}  \{a_1: k_1,  a_2:k_2\}    ]\!] (\mu  ) \models^I \mathbb{P} (\phi) = a_1 b_1 + a_2 b_2 + (a_1 +a_2)b_3  $. This means that $ [\![ \mathbb{P} (  \phi  )]\!]^I_{ [\![    X\xleftarrow[]{\$}  \{a_1: k_1,  a_2:k_2\}         ]\!] \mu } = a_1 b_1 + a_2 b_2 + (a_1 +a_2)b_3 =  [\![  pt  (  X\xleftarrow[]{\$}  \{a_1: k_1,  a_2:k_2\} ,  \mathbb{P} (\phi) ) ]\!]^I_{\mu } $.

 \item SEQ: $   [\![  pt  ( C_1; C_2,  \mathbb{P} (\phi) ) ]\!]^I_{\mu } =   [\![   pt(C_1, pt( C_2, \mathbb{P} (\phi)) )]\!]^I_{\mu } = $\\  $   [\![    pt(C_2, \mathbb{P} (\phi)    )  ]\!]^I_{  [\![  C_1 ]\!]  \mu }  =   [\![ \mathbb{P} (\phi)]\!]^I_{ [\![  C_2 ]\!]  [\![  C_1 ]\!]  \mu }  =    [\![ \mathbb{P} (\phi)  ]\!]^I_{    [\![  C_1 ;C_2 ]\!]  \mu }    $.
 
 \item IF: $[\![  pt  (\texttt{if}\ B\ \texttt{then}\ C_{1}\ \texttt{else}\ C_{2}  ,  \mathbb{P} (\phi) ) ]\!]^I_{\mu } =   $\\

  $ [\![ pt(C_1,   \mathbb{P} (\phi    ) )/B   +   pt(C_2,  \mathbb{P} (\phi  )  )/(\neg B)  ]\!]^I_{\mu } =        $ \\

  $ [\![ pt(C_1,  \mathbb{P} (\phi  )  )/B      ]\!]^I_{\mu }  +      [\![   pt(C_2,  \mathbb{P} (\phi  )  )/(\neg B)  ]\!]^I_{\mu }   = $ \\

    $ [\![ pt(C_1,  \mathbb{P} (\phi  )  )      ]\!]^I_{ \downarrow_{B} \mu }  +      [\![   pt(C_2,  \mathbb{P} (\phi  )  )   ]\!]^I_{ \downarrow_{\neg B} \mu }   = $ \\

  $ [\![    \mathbb{P} (\phi  )       ]\!]^I_{[\![ C_1 ]\!] \downarrow_{B} \mu }  +      [\![      \mathbb{P} (\phi  )    ]\!]^I_{[\![ C_2 ]\!] \downarrow_{\neg B} \mu }   = $ \\
  
  $ [\![    \mathbb{P} (\phi  )       ]\!]^I_{   [\![ C_1 ]\!] \downarrow_{B} \mu +      [\![ C_2 ]\!] \downarrow_{\neg B} \mu  }    =   $   
  $ [\![    \mathbb{P} (\phi  )       ]\!]^I_{   [\![   \texttt{if}\ B\ \texttt{then}\ C_{1}\ \texttt{else}\ C_{2}     ]\!]      \mu }       $

\item WHILE: Without loss of generality, let $sp(\mu ) = \{S_{0,\infty}, S_{0,0},S_{0,1},\ldots  \}$, in which  $S_{0,i} \models wp(i) $ and $S_{0,\infty} \models   wp(\infty )$. Equivalently, we may let $\mu_{S_{0,i}}= \downarrow_{wp(i)} (\mu) $ and $\mu_{S_{0,\infty}}= \downarrow_{wp(\infty)} (\mu) $.

Then  we know  
$\mu(S_{0,i}) =  [\![  \mathbb{P}(   wp(i)  )  ]\!]_{\mu }$, $\mu(S_{0,\infty }) =  [\![  \mathbb{P}(    wp(\infty)  )  ]\!]_{\mu }$.

  That is, $\mu=  [\![  \mathbb{P}(   wp(\infty)  )  ]\!]_{\mu }  \mu_{S_{0,\infty}}  + \displaystyle\sum_{i = 0}^{\infty}  [\![  \mathbb{P}(  wp(i)  )  ]\!]_{\mu } \mu_{S_{0,i}}$.

Then $  [\![   \mathbb{P}(\phi  ) ]\!]_{ [\![ WL  ]\!] \mu }   $

$$   =[\![   \mathbb{P}(\phi  ) ]\!]_{ [\![ WL  ]\!] (  \mu(S_{0,\infty}) \mu_{S_{0,\infty}}  +\displaystyle\sum_{i = 0}^{\infty} \mu(S_{0,i}) \mu_{S_{0,i}  })  }  $$

$$ =[\![   \mathbb{P}(\phi  ) ]\!]_{ [\![ WL ]\!] (  \mu(S_{0,\infty}) \mu_{S_{0,\infty}}    ) } +  [\![   \mathbb{P}(\phi  ) ]\!]_{ [\![WL  ]\!]  \displaystyle\sum_{i = 0}^{\infty} \mu(S_{0,i}) \mu_{S_{0,i}}  } $$

$$=  \mu(S_{0,\infty})  [\![   \mathbb{P}(\phi  ) ]\!]_{ [\![ WL ]\!] (   \mu_{S_{0,\infty }}    ) } +  \displaystyle\sum_{i = 0}^{\infty} \mu(S_{0,i}) [\![   \mathbb{P}(\phi  ) ]\!]_{  [\![ WL  ]\!] \mu_{S_{0,i}}  }  $$

By Lemma \ref{key lemma} we know $[\![ WL  ]\!] \mu_{S_{0,i}}  =   [\![ (IF)^i   ]\!] \mu_{S_{0,i}} $.  \\

Therefore, 

$ [\![ \mathbb{P}(\phi  ) ]\!]_{  [\![WL  ]\!] \mu_{S_{0,i}}  } =   [\![ \mathbb{P}(\phi  ) ]\!]_{  [\![ (IF)^i   ]\!] \mu_{S_{0,i}}    }   $

By induction hypothesis we know
$$  [\![ \mathbb{P}(\phi  ) ]\!]_{  [\![ (  IF)^i   ]\!] \mu_{S_{0,i}}    }  =  [\![pt ( (  IF)^i  ,\mathbb{P}(\phi  )  ) ]\!]_{    \mu_{S_{0,i}}    }  $$

Then we have 
$$  [\![ \mathbb{P}(\phi  ) ]\!]_{ [\![WL  ]\!] \mu_{S_{0,i}}    }  =  [\![pt ( (IF)^i  ,\mathbb{P}(\phi  )  ) ]\!]_{    \mu_{S_{0,i}}    }  $$

Moreover, $$  [\![pt ( (IF)^i  ,\mathbb{P}(\phi  )  ) ]\!]_{    \mu_{S_{0,i}}    } =[\![pt ( (IF)^i  ,\mathbb{P}(\phi  )  ) ]\!]_{    \downarrow_{wp(i)}( \mu)    } $$

$$ =  [\![pt ( (   IF)^i  ,\mathbb{P}(\phi  )  ) / (   wp(i)     ) ]\!]_{    \mu     }  $$

At this stage we know $  [\![   \mathbb{P}(\phi  ) ]\!]_{ [\![WL ]\!] \mu }   = $

$$  \mu(S_{0,\infty})  [\![   \mathbb{P}(\phi  ) ]\!]_{ [\![WL  ]\!] (   \mu_{S_{0,\infty }}    ) } +  \displaystyle\sum_{i = 0}^{\infty} \mu(S_{0,i}) [\![   \mathbb{P}(\phi  ) ]\!]_{  [\![ WL ]\!] \mu_{S_{0,i}}  }  $$

in which $ \displaystyle\sum_{i = 0}^{\infty} \mu(S_{0,i}) [\![   \mathbb{P}(\phi  ) ]\!]_{  [\![ WL  ]\!] \mu_{S_{0,i}}  } = $

$ [\![  \displaystyle\sum_{i = 0}^{\infty} (  \mathbb{P}(    wp(i)    )      ( pt ( (  IF)^i  ,\mathbb{P}(\phi  )  ) / (    wp(i)   )) )  ]\!]_{\mu } $

Note that $SUM$ is short for $$\displaystyle\sum_{i = 0}^{\infty} (  \mathbb{P}(    wp(i)    )      ( pt ( (  IF)^i  ,\mathbb{P}(\phi  )  ) / (    wp(i)   )) ).$$ 

Then $  [\![   \mathbb{P}(\phi  ) ]\!]_{ [\![WL ]\!] \mu }   = $
$$  \mu(S_{0,\infty})  [\![   \mathbb{P}(\phi  ) ]\!]_{ [\![WL  ]\!] (   \mu_{S_{0,\infty }}    ) } + [\![ SUM ]\!]_{\mu} $$
$$=  [\![    \mathbb{P}(wp(\infty)  )    ]\!]_{\mu} [\![   \mathbb{P}(\phi  ) ]\!]_{ [\![WL  ]\!] (   \mu_{S_{0,\infty }}    ) } + [\![ SUM ]\!]_{\mu}  .$$

\ \\

It remains to study $[\![ \mathbb{P}(\phi  ) ]\!]_{  [\![ WL ]\!] \mu_{S_{0,\infty}}  }$.

Note that $ [\![ WL ]\!]  = [\![ IF;WL ]\!]$.
 
 We then know $[\![ \mathbb{P}(\phi  ) ]\!]_{  [\![WL ]\!] \mu_{S_{0,\infty}}  } $\\
 $= [\![ \mathbb{P}(\phi  ) ]\!]_{  [\![ IF; WL  ]\!] \mu_{S_{0,\infty}}  }$\\
  $= [\![ \mathbb{P}(\phi  ) ]\!]_{  [\![  WL  ]\!] [\![ IF ]\!]    \mu_{S_{0,\infty}}  }  $\\
 $= [\![ \mathbb{P}(\phi  ) ]\!]_{  [\![  WL  ]\!] [\![ \texttt{if}\ B\ \texttt{then}\ C \ \texttt{else} \mbox{ }  \texttt{skip}   ]\!]    \mu_{S_{0,\infty}}  }  $\\
  $= [\![ \mathbb{P}(\phi  ) ]\!]_{  [\![  WL  ]\!] [\![   C     ]\!]    \mu_{S_{0,\infty}}  }  $.

 Here we also note that $[\![   C     ]\!]    \mu_{S_{0,\infty}}= [\![   C     ]\!]\downarrow_{wp(\infty)}   ( \mu )$.

Without loss of generality,  let $sp ( [\![   C     ]\!]  \mu_{S_{0,\infty}} ) = \{S_{1,\infty}, S_{1,0},S_{1,1},\ldots,  \} $, in which  $S_{1,i} \models wp(i)  $ and $S_{1,\infty} \models  wp(\infty)$.

Then by repeating the reasoning on the cases of  $S_{0,i}$, we know \\

  $  [\![   \mathbb{P}(\phi  ) ]\!]_{ [\![ WL  ]\!]   [\![   C     ]\!]    \mu_{S_{0,\infty}}   }   = $

$$  [\![    \mathbb{P}(wp(\infty)  )    ]\!]_{[\![   C     ]\!]    \mu_{S_{0,\infty}} } [\![   \mathbb{P}(\phi  ) ]\!]_{ [\![WL  ]\!] (   \mu_{S_{1,\infty }}    ) } + [\![ SUM ]\!]_{[\![   C     ]\!]    \mu_{S_{0,\infty}} }  =$$

$$  [\![    \mathbb{P}(wp(\infty)  )    ]\!]_{[\![   C     ]\!]  \downarrow_{wp(\infty)}  (\mu)  } [\![   \mathbb{P}(\phi  ) ]\!]_{ [\![WL  ]\!] (   \mu_{S_{1,\infty }}    ) } + [\![ SUM ]\!]_{[\![   C     ]\!]    \downarrow_{wp(\infty)} (\mu ) }  $$

By induction hypothesis we know 
$$  [\![ \mathbb{P}(wp(\infty)  )    ]\!]_{[\![   C     ]\!]  \downarrow_{wp(\infty)}  (\mu)  }  = [\![pt(C, \mathbb{P}(wp(\infty)  ))    ]\!]_{   \downarrow_{wp(\infty)}  (\mu)  }$$
$$=  [\![pt(C, \mathbb{P}(wp(\infty)  ))  /wp(\infty)  ]\!]_{     \mu  } $$

and $[\![ SUM ]\!]_{[\![   C     ]\!]    \downarrow_{wp(\infty)} (\mu ) } =  [\![pt(C,  SUM)  /wp(\infty)  ]\!]_{     \mu  } $.

\ \\

It then remains to study $[\![ \mathbb{P}(\phi  ) ]\!]_{  [\![ WL ]\!] \mu_{S_{1,\infty}}  }$.

 By repeating the reasoning on the cases of  $S_{0,i}$, we know \\

  $  [\![   \mathbb{P}(\phi  ) ]\!]_{ [\![ WL  ]\!]   [\![   C     ]\!]    \mu_{S_{1,\infty}}   }   = $
  
  $$  [\![    \mathbb{P}(wp(\infty)  )    ]\!]_{[\![   C     ]\!]    \mu_{S_{1,\infty}} } [\![   \mathbb{P}(\phi  ) ]\!]_{ [\![WL  ]\!] (   \mu_{S_{2,\infty }}    ) } + [\![ SUM ]\!]_{[\![   C     ]\!]    \mu_{S_{1,\infty}} }  $$
  
  in which  $ \mu_{S_{2,\infty }} = \downarrow_{wp(\infty)} ([\![   C     ]\!] ( \mu_{S_{1,\infty }} ) ) $.
  
\ \\
Let $ \mu_{S_{i+1,\infty }} = \downarrow_{wp(\infty)} ([\![   C     ]\!] ( \mu_{S_{i,\infty }} ) ) $.

Repeat the above procedure   to infinity we get $ [\![  pt(WL, \mathbb{P}(\phi  ) )  ]\!]_{\mu }=[\![   \mathbb{P}(\phi   )  ]\!]_{ [\![   WL  ]\!] (\mu) } $

\end{enumerate}

\end{proof}

The weakest preterm calculus and the above characterization lemma of weakest preterms are the main contribution of this article. With these notions and results at hand, we can proceed to define the weakest precondition of probabilistic assertion and the proof system of PHL with probabilistic assertion straightforwardly. 

\begin{definition}[Weakest precondition of probabilistic assertion]\label{def.weakestpreconditionprobabilistic}

\ \\

\begin{enumerate}
\item $WP(C, r_1 \mbox{ }  rop \mbox{ } r_2  ) = pt(C, r_1) \mbox{ }  rop \mbox{ }  pt(C, r_2)$

\item $  WP(C, \neg \Phi) = \neg WP(C, \Phi) $

\item $WP(C,\Phi_1 \wedge \Phi_2) = WP(C,\Phi_1) \wedge WP(C,\Phi_2) $

\end{enumerate}

\end{definition}

\begin{theorem}\label{wp for PA}
 
   $\mu \models^I WP(C, \Phi)$ iff $ [\![C]\!]\mu \models^I \Phi$.

\end{theorem}

\begin{proof}
We prove by structural induction:

\begin{enumerate}
\item   $\mu \models^I WP(C, r_1  \mbox{ }  rop \mbox{ }  r_2)$ iff  $\mu \models^I pt(C, r_1) \mbox{ }  rop \mbox{ }  pt(C, r_2)$ iff \\
 $[\![ pt(C, r_1)  ]\!]^I_{  \mu }  \mbox{ }  rop \mbox{ }   [\![  pt(C, r_2) ]\!]^I_{   \mu }  = \top$ iff\\
 $[\![  r_1  ]\!]^I_{ [\![ C]\!] \mu }   \mbox{ }  rop \mbox{ }    [\![  r_2  ]\!]_{ [\![ C]\!]^I \mu }  = \top$ iff $ [\![ C]\!] \mu \models^I r_1 \mbox{ }  rop \mbox{ } r_2$.

\item $\mu \models^I  WP(C, \neg \Phi)$ iff $ \mu \models^I \neg WP(C, \Phi) $ iff $ \mu   \not\models^I   WP(C, \Phi) $ iff $[\![C]\!]\mu \not\models^I \Phi$ iff $[\![C]\!]\mu  \models^I \neg \Phi$.

\item $\mu \models^I WP(\Phi_1 \wedge \Phi_2)$ iff  $\mu \models^I WP(C,\Phi_1) \wedge WP(C,\Phi_2)$ iff $\mu \models^I WP(C,\Phi_1)  $ and $\mu \models^I   WP(C,\Phi_2)$ iff $[\![C]\!]\mu  \models^I   \Phi_1$ and $[\![C]\!]\mu  \models^I   \Phi_2$ iff $[\![C]\!]\mu  \models^I   \Phi_1 \wedge \Phi_2$.

\end{enumerate}

\end{proof}

The last theorem suggests we define a proof system of PHL with probabilistic formulas in a uniform manner: $\vdash \{WP(C, \Phi)\} C \{\Phi\}$ for every command $C$ and probabilistic formula $\Phi$. We will follow the suggestion with some minor modifications.

\begin{definition}[Proof system of probabilistic formula PHL]    The proof system of  PHL with probabilistic assertions   consists of the following inference rules:

 \begin{center}
        \begin{tabular}{ll}
         $SKIP:$ & $\frac{}{\vdash\{  \Phi    \}\texttt{skip}\{  \Phi  \}}$\\
         $AS:$ & $ \frac{}{\vdash \{      WP(X \leftarrow E, \mbox{ }\Phi)   \}  X \leftarrow E \{  \Phi \} }$\\
         $PAS:$ & $   \frac{}{\vdash \{  WP(X\xleftarrow{\$}\{a_1: r_{1},...,a_n: r_{n} \},\mbox{ } \Phi)    \}  X\xleftarrow{\$}\{a_1: r_{1},...,a_n: r_{n} \}   \{    \Phi    \} }$\\  
         \\
         $SEQ:$ & $ \frac{\vdash\{\Phi\}C_{1}\{\Phi_{1}\}\quad \vdash\{\Phi_{1}\}C_{2}\{\Phi_{2}\}} {\vdash\{\Phi\}C_{1};C_{2}\{\Phi_{2}\}} $\\
     
           $IF:$ & $\frac{}{\vdash\{  WP( \texttt{if}\ B\ \texttt{then}\ C_{1}\ \texttt{else}\ C_{2}  , \mbox{ } \Phi)    \}  \texttt{if}\ B\ \texttt{then}\ C_{1}\ \texttt{else}\ C_{2}   \{  \Phi \}}$\\
 
 $WHILE:$ & $\frac{}{\vdash\{   WP( \texttt{while} \mbox{ }B \mbox{ } \texttt{do} \mbox{ }C , \mbox{ } \Phi)   \}  \texttt{while} \mbox{ }B \mbox{ } \texttt{do} \mbox{ }C   \{ \Phi \}}$,\\
                 
                 \\

          $CONS:$ & $ \frac{\models\Phi'\rightarrow\Phi\quad  \vdash\{\Phi\}C\{\Psi\}\quad \models\Psi\rightarrow\Psi'}  {\vdash\{\Phi'\}C\{\Psi'\}}$\\

        \end{tabular}
    \end{center}

\end{definition}

With the rule (CONS) in our proof system, we can treat $WP(C,\Phi)$ as a semantic notion: if $\models WP(C,\Phi) \leftrightarrow \Psi$, then $\Psi$ is conceived as the weakest precondition of $\Phi$ with the command $C$.

\begin{example}

$\vdash \{\top\}  \texttt{while} \mbox{ } \top \mbox{ } \texttt{do} \mbox{ } \texttt{skip}  \{ \mathbb{P}(\top  ) =0  \}$ is derivable in our proof system. This is because $pt(  \texttt{while} \mbox{ } \top \mbox{ } \texttt{do} \mbox{ } \texttt{skip} , \mathbb{P}(\top  )      )=0$ and $pt(  \texttt{while} \mbox{ } \top \mbox{ } \texttt{do} \mbox{ } \texttt{skip} , 0 )=0$, which implies that \\
$  WP(\texttt{while} \mbox{ } \top \mbox{ } \texttt{do} \mbox{ } \texttt{skip} , \mathbb{P}(\top  ) =0 ) =  (  0=0) $. By the (While) rule we derive $\vdash \{  0=0 \}  \texttt{while} \mbox{ } \top \mbox{ } \texttt{do} \mbox{ } \texttt{skip}  \{ \mathbb{P}(\top  ) =0  \}$, then by the (CONS) rule we derive $\vdash \{\top\}  \texttt{while} \mbox{ } \top \mbox{ } \texttt{do} \mbox{ } \texttt{skip}  \{ \mathbb{P}(\top  ) =0  \}$.

\end{example}

\begin{example}

 Let $C$ be the following program: 
$$X \leftarrow \{\frac{1}{3} :0, \frac{2}{3} :1\}  ; \texttt{if}\ X = 0 \ \texttt{then}\  (\texttt{while} \mbox{ } \top \mbox{ } \texttt{do} \mbox{ } \texttt{skip})   \ \texttt{else}\ \texttt{skip}   $$

$\vdash \{\top\}  C  \{ \mathbb{P}( \top  )  \leq  \frac{2}{3}  \}$ is derivable in our proof system. This is because $pt(C, \mathbb{P}( \top  )) = \frac{2}{3} \mathbb{P}( \top  )$.  Then we know $WP(C, \mathbb{P}( \top  ) \leq \frac{2}{3}) =  \frac{2}{3} \mathbb{P}( \top  ) \leq  \frac{2}{3} $. Then by applying rules in our proof system we get $\vdash \{   \frac{2}{3} \mathbb{P}( \top  ) \leq  \frac{2}{3}      \}  C  \{ \mathbb{P}( \top  )  \leq  \frac{2}{3}  \}$.
Note that $\models    \top \rightarrow ( \frac{2}{3} \mathbb{P}( \top  ) \leq  \frac{2}{3}  ) $. We then use (CONS) to derive $\vdash \{\top\}  C  \{ \mathbb{P}( \top  )  \leq  \frac{2}{3}  \}$.

\end{example}

\begin{lemma}\label{wp condition prob}
If $\models  \{ \Phi \} C \{ \Psi \} $, then $\models \Phi \rightarrow WP(C, \Psi)$.
\end{lemma}

\begin{proof}
Let $\mu$ be a probabilistic state such that $\mu \models \Phi$. Then $ [\![ C]\!] \mu \models \Psi $. Then by Theorem \ref{wp for PA} we know $\mu \models WP(C,\Psi)$.
\end{proof}

\begin{lemma}
   For an arbitrary command $C$ and an arbitrary probabilistic formula $\Phi$, it holds that $\vdash \{WP(C,\Phi)\}C\{\Phi\}$.
\end{lemma}

\begin{proof}

We prove it by induction on the structure of command $C$.
When $C$ is \texttt{skip}, assignment, random assignment, IF or WHILE command, we can directly derive them from the proof system PHL. The only case that needs to show is the sequential command, which means that we need to prove $\vdash \{WP(C_1;C_2,\Phi)\}C_1;C_2\{\Phi\}$.

By Definition~\ref{def.weakestpreconditionprobabilistic}, we have $WP(C_1;C_2,\Phi)=WP(C_1,WP(C_2,\Phi))$. By induction hypothesis,  it holds that $\vdash\{WP(C_2,\Phi)\}C_2\{\Phi\}$ and $\vdash\{WP(C_1,WP(C_2,\Phi))\}C_1\{WP(C_2;\Phi)\}$. By inference rule $SEQ$, we can conclude that $\vdash\{WP(C_1,WP(C_2,\Phi))\}C_1;C_2\{\Phi\}$ which implies that $\vdash \{WP(C_1;C_2,\Phi)\}C_1;C_2\{\Phi\}$.

\end{proof}

\vspace{-2mm}

\vspace{-2mm}
\begin{theorem}

The proof system of  PHL with probabilistic assertions is sound and complete.

\end{theorem}
\vspace{-2mm}
\begin{proof}
Soundness: the soundness of SKIP, AS, PAS, IF, WHILE follows from Theorem \ref{wp for PA}. The soundness of SEQ and CON can be proved by simple deduction.

Completeness: Assume $\models  \{ \Phi \} C \{ \Psi \} $, then $\models \Phi \rightarrow WP(C, \Psi)$ by lemma \ref{wp condition prob}. Then by $\vdash \{WP(C,\Psi) \}C \{ \Psi \} $ and CONS we know $ \vdash  \{\Phi \}C \{ \Psi \} $.

\end{proof}

\vspace{-2mm}
\section{Conclusions and Future Work}\label{Conclusions and Future Work}
Probabilistic Hoare logic is particularly useful in fields like the formal verification of probabilistic programs. Developers and researchers can use it to formally reason about program behavior with stochastic elements and ensure that the desired probabilistic properties are upheld. The studies on Hoare logic can be classified into two approaches: satisfaction-based and expectation-based. The problem of relative completeness of satisfaction-based PHL with While-loop has been unsolved since 1979. This paper addresses this problem by proposing a new PHL system that introduces the command of probabilistic assignment. In comparison with the existing literature on satisfaction-based PHL where probabilistic choice is expressed in either probabilistic choice between two statements or a flipped coin, our construction is expressively equivalent in terms of programming language and, moreover, it brings a lot of convenience in defining the weakest preconditions and preterms as an extension to the normal assignments. This, in turn, facilitates the proof of relative completeness. 

The main contribution of this paper is successfully finding out the weakest preterm of While-loop given a real expression. Definition~\ref{def.weakestpreterm} (9) essentially reveals how a While-loop changes the probabilistic property of computer states, considering both execution branches that halt and infinite runs.  Lemma~\ref{preterm lemma} demonstrates that the weakest preterms defined in our way accurately characterize the `pre-probability' given a command $C$ and a real expression $r$. The appropriate weakest preterm calculus bridges the biggest gap in proving the relative completeness of PHL.


The progress we have made in this paper may shed insights into the research of quantum Hoare logic (QHL). In the expectation-based QHL \cite{Ying11,LiuZWYLLYZ19}, the assertions in the Hoare triples are functions which map states to observables. They are interpreted as real numbers by calculating the traces of some matrix representing observing quantum states. In this approach, it is difficult to handle compounded properties of quantum programs. On the other hand, the assertions in the satisfaction-based approach \cite{ChadhaMS06,Kakutani09,Unruh19,ZhouYY19,DengF22} are logical formulas. This treatment makes it easier to express the properties of computer states, making the formal verification of programs more straightforward. However, the existing satisfaction-based QHL is either incomplete or complete but not expressive enough. Therefore, in the future, we will study how to extend our current work to build satisfaction-based QHL that is both complete and expressively strong.

\newpage

\bibliographystyle{ACM-Reference-Format}
\bibliography{literature}


\begin{thebibliography}{35}


\ifx \showCODEN    \undefined \def \showCODEN     #1{\unskip}     \fi
\ifx \showDOI      \undefined \def \showDOI       #1{#1}\fi
\ifx \showISBNx    \undefined \def \showISBNx     #1{\unskip}     \fi
\ifx \showISBNxiii \undefined \def \showISBNxiii  #1{\unskip}     \fi
\ifx \showISSN     \undefined \def \showISSN      #1{\unskip}     \fi
\ifx \showLCCN     \undefined \def \showLCCN      #1{\unskip}     \fi
\ifx \shownote     \undefined \def \shownote      #1{#1}          \fi
\ifx \showarticletitle \undefined \def \showarticletitle #1{#1}   \fi
\ifx \showURL      \undefined \def \showURL       {\relax}        \fi
\providecommand\bibfield[2]{#2}
\providecommand\bibinfo[2]{#2}
\providecommand\natexlab[1]{#1}
\providecommand\showeprint[2][]{arXiv:#2}

\bibitem[Apt(1984)]%
        {Apt84}
\bibfield{author}{\bibinfo{person}{Krzysztof~R. Apt}.}
  \bibinfo{year}{1984}\natexlab{}.
\newblock \showarticletitle{Ten Years of Hoare's Logic: {A} Survey Part {II:}
  Nondeterminism}.
\newblock \bibinfo{journal}{\emph{Theor. Comput. Sci.}}  \bibinfo{volume}{28}
  (\bibinfo{year}{1984}), \bibinfo{pages}{83--109}.
\newblock
\urldef\tempurl%
\url{https://doi.org/10.1016/0304-3975(83)90066-X}
\showDOI{\tempurl}


\bibitem[Apt et~al\mbox{.}(2009a)]%
        {apt2009verification}
\bibfield{author}{\bibinfo{person}{Krzysztof~R Apt}, \bibinfo{person}{Frank~S.
  Boer}, {and} \bibinfo{person}{Ernst-R{\"u}diger Olderog}.}
  \bibinfo{year}{2009}\natexlab{a}.
\newblock \bibinfo{booktitle}{\emph{Verification of sequential and concurrent
  programs}}. Vol.~\bibinfo{volume}{2}.
\newblock \bibinfo{publisher}{Springer}.
\newblock


\bibitem[Apt et~al\mbox{.}(2009b)]%
        {AptBO09}
\bibfield{author}{\bibinfo{person}{Krzysztof~R. Apt}, \bibinfo{person}{Frank~S.
  de Boer}, {and} \bibinfo{person}{Ernst{-}R{\"{u}}diger Olderog}.}
  \bibinfo{year}{2009}\natexlab{b}.
\newblock \bibinfo{booktitle}{\emph{Verification of Sequential and Concurrent
  Programs}}.
\newblock \bibinfo{publisher}{Springer}.
\newblock
\showISBNx{978-1-84882-744-8}
\urldef\tempurl%
\url{https://doi.org/10.1007/978-1-84882-745-5}
\showDOI{\tempurl}


\bibitem[Apt and Olderog(2019)]%
        {apt2019fifty}
\bibfield{author}{\bibinfo{person}{Krzysztof~R Apt} {and}
  \bibinfo{person}{Ernst-R{\"u}diger Olderog}.}
  \bibinfo{year}{2019}\natexlab{}.
\newblock \showarticletitle{Fifty years of Hoare’s logic}.
\newblock \bibinfo{journal}{\emph{Formal Aspects of Computing}}
  \bibinfo{volume}{31} (\bibinfo{year}{2019}), \bibinfo{pages}{751--807}.
\newblock


\bibitem[Barthe et~al\mbox{.}(2013)]%
        {BartheDGKSS13}
\bibfield{author}{\bibinfo{person}{Gilles Barthe},
  \bibinfo{person}{Fran{\c{c}}ois Dupressoir}, \bibinfo{person}{Benjamin
  Gr{\'{e}}goire}, \bibinfo{person}{C{\'{e}}sar Kunz},
  \bibinfo{person}{Benedikt Schmidt}, {and} \bibinfo{person}{Pierre{-}Yves
  Strub}.} \bibinfo{year}{2013}\natexlab{}.
\newblock \showarticletitle{EasyCrypt: {A} Tutorial}. In
  \bibinfo{booktitle}{\emph{Foundations of Security Analysis and Design {VII} -
  {FOSAD} 2012/2013 Tutorial Lectures}} \emph{(\bibinfo{series}{Lecture Notes
  in Computer Science}, Vol.~\bibinfo{volume}{8604})},
  \bibfield{editor}{\bibinfo{person}{Alessandro Aldini},
  \bibinfo{person}{Javier L{\'{o}}pez}, {and} \bibinfo{person}{Fabio
  Martinelli}} (Eds.). \bibinfo{publisher}{Springer},
  \bibinfo{pages}{146--166}.
\newblock
\urldef\tempurl%
\url{https://doi.org/10.1007/978-3-319-10082-1\_6}
\showDOI{\tempurl}


\bibitem[Barthe et~al\mbox{.}(2009)]%
        {BartheGB09}
\bibfield{author}{\bibinfo{person}{Gilles Barthe}, \bibinfo{person}{Benjamin
  Gr{\'{e}}goire}, {and} \bibinfo{person}{Santiago~Zanella B{\'{e}}guelin}.}
  \bibinfo{year}{2009}\natexlab{}.
\newblock \showarticletitle{Formal certification of code-based cryptographic
  proofs}. In \bibinfo{booktitle}{\emph{Proceedings of the 36th {ACM}
  {SIGPLAN-SIGACT} Symposium on Principles of Programming Languages, {POPL}
  2009, Savannah, GA, USA, January 21-23, 2009}},
  \bibfield{editor}{\bibinfo{person}{Zhong Shao} {and}
  \bibinfo{person}{Benjamin~C. Pierce}} (Eds.). \bibinfo{publisher}{{ACM}},
  \bibinfo{pages}{90--101}.
\newblock
\urldef\tempurl%
\url{https://doi.org/10.1145/1480881.1480894}
\showDOI{\tempurl}


\bibitem[Barthe et~al\mbox{.}(2012)]%
        {BartheGB12}
\bibfield{author}{\bibinfo{person}{Gilles Barthe}, \bibinfo{person}{Benjamin
  Gr{\'{e}}goire}, {and} \bibinfo{person}{Santiago~Zanella B{\'{e}}guelin}.}
  \bibinfo{year}{2012}\natexlab{}.
\newblock \showarticletitle{Probabilistic Relational Hoare Logics for
  Computer-Aided Security Proofs}. In \bibinfo{booktitle}{\emph{Mathematics of
  Program Construction - 11th International Conference, {MPC} 2012, Madrid,
  Spain, June 25-27, 2012. Proceedings}} \emph{(\bibinfo{series}{Lecture Notes
  in Computer Science}, Vol.~\bibinfo{volume}{7342})},
  \bibfield{editor}{\bibinfo{person}{Jeremy Gibbons} {and}
  \bibinfo{person}{Pablo Nogueira}} (Eds.). \bibinfo{publisher}{Springer},
  \bibinfo{pages}{1--6}.
\newblock
\urldef\tempurl%
\url{https://doi.org/10.1007/978-3-642-31113-0\_1}
\showDOI{\tempurl}


\bibitem[Batz et~al\mbox{.}(2021)]%
        {BatzKKM21}
\bibfield{author}{\bibinfo{person}{Kevin Batz},
  \bibinfo{person}{Benjamin~Lucien Kaminski}, \bibinfo{person}{Joost{-}Pieter
  Katoen}, {and} \bibinfo{person}{Christoph Matheja}.}
  \bibinfo{year}{2021}\natexlab{}.
\newblock \showarticletitle{Relatively complete verification of probabilistic
  programs: an expressive language for expectation-based reasoning}.
\newblock \bibinfo{journal}{\emph{Proc. {ACM} Program. Lang.}}
  \bibinfo{volume}{5}, \bibinfo{number}{{POPL}} (\bibinfo{year}{2021}),
  \bibinfo{pages}{1--30}.
\newblock
\urldef\tempurl%
\url{https://doi.org/10.1145/3434320}
\showDOI{\tempurl}


\bibitem[Chadha et~al\mbox{.}(2007)]%
        {chadha2007reasoning}
\bibfield{author}{\bibinfo{person}{Rohit Chadha}, \bibinfo{person}{Lu{\'\i}s
  Cruz-Filipe}, \bibinfo{person}{Paulo Mateus}, {and}
  \bibinfo{person}{Am{\'\i}lcar Sernadas}.} \bibinfo{year}{2007}\natexlab{}.
\newblock \showarticletitle{Reasoning about probabilistic sequential programs}.
\newblock \bibinfo{journal}{\emph{Theoretical Computer Science}}
  \bibinfo{volume}{379}, \bibinfo{number}{1-2} (\bibinfo{year}{2007}),
  \bibinfo{pages}{142--165}.
\newblock


\bibitem[Chadha et~al\mbox{.}(2006)]%
        {ChadhaMS06}
\bibfield{author}{\bibinfo{person}{Rohit Chadha}, \bibinfo{person}{Paulo
  Mateus}, {and} \bibinfo{person}{Am{\'{\i}}lcar Sernadas}.}
  \bibinfo{year}{2006}\natexlab{}.
\newblock \showarticletitle{Reasoning About Imperative Quantum Programs}. In
  \bibinfo{booktitle}{\emph{Proceedings of the 22nd Annual Conference on
  Mathematical Foundations of Programming Semantics, {MFPS} 2006, Genova,
  Italy, May 23-27, 2006}} \emph{(\bibinfo{series}{Electronic Notes in
  Theoretical Computer Science}, Vol.~\bibinfo{volume}{158})},
  \bibfield{editor}{\bibinfo{person}{Stephen~D. Brookes} {and}
  \bibinfo{person}{Michael~W. Mislove}} (Eds.). \bibinfo{publisher}{Elsevier},
  \bibinfo{pages}{19--39}.
\newblock
\urldef\tempurl%
\url{https://doi.org/10.1016/J.ENTCS.2006.04.003}
\showDOI{\tempurl}


\bibitem[Corin and den Hartog(2005)]%
        {Hartog05}
\bibfield{author}{\bibinfo{person}{Ricardo Corin} {and} \bibinfo{person}{Jerry
  den Hartog}.} \bibinfo{year}{2005}\natexlab{}.
\newblock \showarticletitle{A Probabilistic Hoare-style logic for Game-based
  Cryptographic Proofs (Extended Version)}.
\newblock \bibinfo{journal}{\emph{Cryptology ePrint Archive}}
  (\bibinfo{year}{2005}).
\newblock


\bibitem[den Hartog(2008)]%
        {Hartog08}
\bibfield{author}{\bibinfo{person}{Jerry den Hartog}.}
  \bibinfo{year}{2008}\natexlab{}.
\newblock \showarticletitle{Towards mechanized correctness proofs for
  cryptographic algorithms: Axiomatization of a probabilistic Hoare style
  logic}.
\newblock \bibinfo{journal}{\emph{Science of Computer Programming}}
  \bibinfo{volume}{74}, \bibinfo{number}{1-2} (\bibinfo{year}{2008}),
  \bibinfo{pages}{52--63}.
\newblock


\bibitem[Den~Hartog and de~Vink(2002)]%
        {Hartog02}
\bibfield{author}{\bibinfo{person}{Jerry Den~Hartog} {and}
  \bibinfo{person}{Erik~P de Vink}.} \bibinfo{year}{2002}\natexlab{}.
\newblock \showarticletitle{Verifying probabilistic programs using a Hoare like
  logic}.
\newblock \bibinfo{journal}{\emph{International journal of foundations of
  computer science}} \bibinfo{volume}{13}, \bibinfo{number}{03}
  (\bibinfo{year}{2002}), \bibinfo{pages}{315--340}.
\newblock


\bibitem[Deng and Feng(2022)]%
        {DengF22}
\bibfield{author}{\bibinfo{person}{Yuxin Deng} {and} \bibinfo{person}{Yuan
  Feng}.} \bibinfo{year}{2022}\natexlab{}.
\newblock \showarticletitle{Formal semantics of a classical-quantum language}.
\newblock \bibinfo{journal}{\emph{Theor. Comput. Sci.}}  \bibinfo{volume}{913}
  (\bibinfo{year}{2022}), \bibinfo{pages}{73--93}.
\newblock
\urldef\tempurl%
\url{https://doi.org/10.1016/J.TCS.2022.02.017}
\showDOI{\tempurl}


\bibitem[Dijkstra(1975)]%
        {Dijkstra75}
\bibfield{author}{\bibinfo{person}{Edsger~W. Dijkstra}.}
  \bibinfo{year}{1975}\natexlab{}.
\newblock \showarticletitle{Guarded Commands, Nondeterminacy and Formal
  Derivation of Programs}.
\newblock \bibinfo{journal}{\emph{Commun. {ACM}}} \bibinfo{volume}{18},
  \bibinfo{number}{8} (\bibinfo{year}{1975}), \bibinfo{pages}{453--457}.
\newblock
\urldef\tempurl%
\url{https://doi.org/10.1145/360933.360975}
\showDOI{\tempurl}


\bibitem[Dijkstra(1976)]%
        {Dijkstra76}
\bibfield{author}{\bibinfo{person}{Edsger~W. Dijkstra}.}
  \bibinfo{year}{1976}\natexlab{}.
\newblock \bibinfo{booktitle}{\emph{A Discipline of Programming}}.
\newblock \bibinfo{publisher}{Prentice-Hall}.
\newblock
\showISBNx{013215871X}
\urldef\tempurl%
\url{https://www.worldcat.org/oclc/01958445}
\showURL{%
\tempurl}


\bibitem[Foley and Hoare(1971)]%
        {FoleyH71}
\bibfield{author}{\bibinfo{person}{M. Foley} {and} \bibinfo{person}{C.~A.~R.
  Hoare}.} \bibinfo{year}{1971}\natexlab{}.
\newblock \showarticletitle{Proof of a Recursive Program: Quicksort}.
\newblock \bibinfo{journal}{\emph{Comput. J.}} \bibinfo{volume}{14},
  \bibinfo{number}{4} (\bibinfo{year}{1971}), \bibinfo{pages}{391--395}.
\newblock
\urldef\tempurl%
\url{https://doi.org/10.1093/COMJNL/14.4.391}
\showDOI{\tempurl}


\bibitem[Gordon et~al\mbox{.}(2014)]%
        {gordon2014probabilistic}
\bibfield{author}{\bibinfo{person}{Andrew~D Gordon}, \bibinfo{person}{Thomas~A
  Henzinger}, \bibinfo{person}{Aditya~V Nori}, {and} \bibinfo{person}{Sriram~K
  Rajamani}.} \bibinfo{year}{2014}\natexlab{}.
\newblock \showarticletitle{Probabilistic programming}.
\newblock In \bibinfo{booktitle}{\emph{Future of Software Engineering
  Proceedings}}. \bibinfo{pages}{167--181}.
\newblock


\bibitem[Hoare(1969)]%
        {hoare1969axiomatic}
\bibfield{author}{\bibinfo{person}{Charles Antony~Richard Hoare}.}
  \bibinfo{year}{1969}\natexlab{}.
\newblock \showarticletitle{An axiomatic basis for computer programming}.
\newblock \bibinfo{journal}{\emph{Commun. ACM}} \bibinfo{volume}{12},
  \bibinfo{number}{10} (\bibinfo{year}{1969}), \bibinfo{pages}{576--580}.
\newblock


\bibitem[Hoare(1971a)]%
        {hoare1971procedures}
\bibfield{author}{\bibinfo{person}{Charles Antony~Richard Hoare}.}
  \bibinfo{year}{1971}\natexlab{a}.
\newblock \showarticletitle{Procedures and parameters: An axiomatic approach}.
  In \bibinfo{booktitle}{\emph{Proceedings of symposium on the semantics of
  algorithmic languages. Lecture notes in mathematics 188}}. Springer,
  \bibinfo{pages}{102--116}.
\newblock


\bibitem[Hoare(1971b)]%
        {Hoare71}
\bibfield{author}{\bibinfo{person}{C.~A.~R. Hoare}.}
  \bibinfo{year}{1971}\natexlab{b}.
\newblock \showarticletitle{Procedures and parameters: An axiomatic approach}.
\newblock In \bibinfo{booktitle}{\emph{Symposium on Semantics of Algorithmic
  Languages}}, \bibfield{editor}{\bibinfo{person}{Erwin Engeler}} (Ed.).
  \bibinfo{series}{Lecture Notes in Mathematics}, Vol.~\bibinfo{volume}{188}.
  \bibinfo{publisher}{Springer}, \bibinfo{pages}{102--116}.
\newblock
\urldef\tempurl%
\url{https://doi.org/10.1007/BFB0059696}
\showDOI{\tempurl}


\bibitem[Jones(1990)]%
        {Jones90}
\bibfield{author}{\bibinfo{person}{C. Jones}.} \bibinfo{year}{1990}\natexlab{}.
\newblock \emph{\bibinfo{title}{Probabilistic Nondeterminism}}.
\newblock \bibinfo{thesistype}{Ph.\,D. Dissertation}.
  \bibinfo{school}{University of Edinburgh}.
\newblock


\bibitem[Kakutani(2009)]%
        {Kakutani09}
\bibfield{author}{\bibinfo{person}{Yoshihiko Kakutani}.}
  \bibinfo{year}{2009}\natexlab{}.
\newblock \showarticletitle{A Logic for Formal Verification of Quantum
  Programs}. In \bibinfo{booktitle}{\emph{Advances in Computer Science -
  {ASIAN} 2009. Information Security and Privacy, 13th Asian Computing Science
  Conference, Seoul, Korea, December 14-16, 2009. Proceedings}}
  \emph{(\bibinfo{series}{Lecture Notes in Computer Science},
  Vol.~\bibinfo{volume}{5913})}, \bibfield{editor}{\bibinfo{person}{Anupam
  Datta}} (Ed.). \bibinfo{publisher}{Springer}, \bibinfo{pages}{79--93}.
\newblock
\urldef\tempurl%
\url{https://doi.org/10.1007/978-3-642-10622-4\_7}
\showDOI{\tempurl}


\bibitem[Kozen(1985)]%
        {Kozen85}
\bibfield{author}{\bibinfo{person}{Dexter Kozen}.}
  \bibinfo{year}{1985}\natexlab{}.
\newblock \showarticletitle{A Probabilistic {PDL}}.
\newblock \bibinfo{journal}{\emph{J. Comput. Syst. Sci.}} \bibinfo{volume}{30},
  \bibinfo{number}{2} (\bibinfo{year}{1985}), \bibinfo{pages}{162--178}.
\newblock
\urldef\tempurl%
\url{https://doi.org/10.1016/0022-0000(85)90012-1}
\showDOI{\tempurl}


\bibitem[Liu et~al\mbox{.}(2019)]%
        {LiuZWYLLYZ19}
\bibfield{author}{\bibinfo{person}{Junyi Liu}, \bibinfo{person}{Bohua Zhan},
  \bibinfo{person}{Shuling Wang}, \bibinfo{person}{Shenggang Ying},
  \bibinfo{person}{Tao Liu}, \bibinfo{person}{Yangjia Li},
  \bibinfo{person}{Mingsheng Ying}, {and} \bibinfo{person}{Naijun Zhan}.}
  \bibinfo{year}{2019}\natexlab{}.
\newblock \showarticletitle{Formal Verification of Quantum Algorithms Using
  Quantum Hoare Logic}. In \bibinfo{booktitle}{\emph{Computer Aided
  Verification - 31st International Conference, {CAV} 2019, New York City, NY,
  USA, July 15-18, 2019, Proceedings, Part {II}}}
  \emph{(\bibinfo{series}{Lecture Notes in Computer Science},
  Vol.~\bibinfo{volume}{11562})}, \bibfield{editor}{\bibinfo{person}{Isil
  Dillig} {and} \bibinfo{person}{Serdar Tasiran}} (Eds.).
  \bibinfo{publisher}{Springer}, \bibinfo{pages}{187--207}.
\newblock
\urldef\tempurl%
\url{https://doi.org/10.1007/978-3-030-25543-5\_12}
\showDOI{\tempurl}


\bibitem[Morgan and McIver(1999)]%
        {Morgan99}
\bibfield{author}{\bibinfo{person}{Carroll Morgan} {and}
  \bibinfo{person}{Annabelle McIver}.} \bibinfo{year}{1999}\natexlab{}.
\newblock \showarticletitle{pGCL: formal reasoning for random algorithms}.
\newblock \bibinfo{journal}{\emph{South African Computer Journal}}
  (\bibinfo{year}{1999}).
\newblock


\bibitem[Morgan et~al\mbox{.}(1996)]%
        {Morgan96}
\bibfield{author}{\bibinfo{person}{Carroll Morgan}, \bibinfo{person}{Annabelle
  McIver}, {and} \bibinfo{person}{Karen Seidel}.}
  \bibinfo{year}{1996}\natexlab{}.
\newblock \showarticletitle{Probabilistic Predicate Transformers}.
\newblock \bibinfo{journal}{\emph{{ACM} Trans. Program. Lang. Syst.}}
  \bibinfo{volume}{18}, \bibinfo{number}{3} (\bibinfo{year}{1996}),
  \bibinfo{pages}{325--353}.
\newblock
\urldef\tempurl%
\url{https://doi.org/10.1145/229542.229547}
\showDOI{\tempurl}


\bibitem[Ramshaw(1979)]%
        {Ramshaw79}
\bibfield{author}{\bibinfo{person}{L.H. Ramshaw}.}
  \bibinfo{year}{1979}\natexlab{}.
\newblock \emph{\bibinfo{title}{Formalizing the analysis of algorithms}}.
\newblock \bibinfo{thesistype}{Ph.\,D. Dissertation}. \bibinfo{school}{Stanford
  University}.
\newblock


\bibitem[Rand and Zdancewic(2015)]%
        {rand2015vphl}
\bibfield{author}{\bibinfo{person}{Robert Rand} {and} \bibinfo{person}{Steve
  Zdancewic}.} \bibinfo{year}{2015}\natexlab{}.
\newblock \showarticletitle{VPHL: A verified partial-correctness logic for
  probabilistic programs}.
\newblock \bibinfo{journal}{\emph{Electronic Notes in Theoretical Computer
  Science}}  \bibinfo{volume}{319} (\bibinfo{year}{2015}),
  \bibinfo{pages}{351--367}.
\newblock


\bibitem[Srivastava et~al\mbox{.}(2014)]%
        {SrivastavaHKSS14}
\bibfield{author}{\bibinfo{person}{Nitish Srivastava},
  \bibinfo{person}{Geoffrey~E. Hinton}, \bibinfo{person}{Alex Krizhevsky},
  \bibinfo{person}{Ilya Sutskever}, {and} \bibinfo{person}{Ruslan
  Salakhutdinov}.} \bibinfo{year}{2014}\natexlab{}.
\newblock \showarticletitle{Dropout: a simple way to prevent neural networks
  from overfitting}.
\newblock \bibinfo{journal}{\emph{J. Mach. Learn. Res.}} \bibinfo{volume}{15},
  \bibinfo{number}{1} (\bibinfo{year}{2014}), \bibinfo{pages}{1929--1958}.
\newblock
\urldef\tempurl%
\url{https://doi.org/10.5555/2627435.2670313}
\showDOI{\tempurl}


\bibitem[Sutskever et~al\mbox{.}(2013)]%
        {SutskeverMDH13}
\bibfield{author}{\bibinfo{person}{Ilya Sutskever}, \bibinfo{person}{James
  Martens}, \bibinfo{person}{George~E. Dahl}, {and}
  \bibinfo{person}{Geoffrey~E. Hinton}.} \bibinfo{year}{2013}\natexlab{}.
\newblock \showarticletitle{On the importance of initialization and momentum in
  deep learning}. In \bibinfo{booktitle}{\emph{Proceedings of the 30th
  International Conference on Machine Learning, {ICML} 2013, Atlanta, GA, USA,
  16-21 June 2013}} \emph{(\bibinfo{series}{{JMLR} Workshop and Conference
  Proceedings}, Vol.~\bibinfo{volume}{28})}. \bibinfo{publisher}{JMLR.org},
  \bibinfo{pages}{1139--1147}.
\newblock
\urldef\tempurl%
\url{http://proceedings.mlr.press/v28/sutskever13.html}
\showURL{%
\tempurl}


\bibitem[Unruh(2019)]%
        {Unruh19}
\bibfield{author}{\bibinfo{person}{Dominique Unruh}.}
  \bibinfo{year}{2019}\natexlab{}.
\newblock \showarticletitle{Quantum relational Hoare logic}.
\newblock \bibinfo{journal}{\emph{Proc. {ACM} Program. Lang.}}
  \bibinfo{volume}{3}, \bibinfo{number}{{POPL}} (\bibinfo{year}{2019}),
  \bibinfo{pages}{33:1--33:31}.
\newblock
\urldef\tempurl%
\url{https://doi.org/10.1145/3290346}
\showDOI{\tempurl}


\bibitem[Winskel(1993)]%
        {Winskel93}
\bibfield{author}{\bibinfo{person}{Glynn Winskel}.}
  \bibinfo{year}{1993}\natexlab{}.
\newblock \bibinfo{booktitle}{\emph{The Formal Semantics of Programming
  Languages}}.
\newblock \bibinfo{publisher}{The MIT Press}.
\newblock


\bibitem[Ying(2011)]%
        {Ying11}
\bibfield{author}{\bibinfo{person}{Mingsheng Ying}.}
  \bibinfo{year}{2011}\natexlab{}.
\newblock \showarticletitle{Floyd-hoare logic for quantum programs}.
\newblock \bibinfo{journal}{\emph{{ACM} Trans. Program. Lang. Syst.}}
  \bibinfo{volume}{33}, \bibinfo{number}{6} (\bibinfo{year}{2011}),
  \bibinfo{pages}{19:1--19:49}.
\newblock
\urldef\tempurl%
\url{https://doi.org/10.1145/2049706.2049708}
\showDOI{\tempurl}


\bibitem[Zhou et~al\mbox{.}(2019)]%
        {ZhouYY19}
\bibfield{author}{\bibinfo{person}{Li Zhou}, \bibinfo{person}{Nengkun Yu},
  {and} \bibinfo{person}{Mingsheng Ying}.} \bibinfo{year}{2019}\natexlab{}.
\newblock \showarticletitle{An applied quantum Hoare logic}. In
  \bibinfo{booktitle}{\emph{Proceedings of the 40th {ACM} {SIGPLAN} Conference
  on Programming Language Design and Implementation, {PLDI} 2019, Phoenix, AZ,
  USA, June 22-26, 2019}}, \bibfield{editor}{\bibinfo{person}{Kathryn~S.
  McKinley} {and} \bibinfo{person}{Kathleen Fisher}} (Eds.).
  \bibinfo{publisher}{{ACM}}, \bibinfo{pages}{1149--1162}.
\newblock
\urldef\tempurl%
\url{https://doi.org/10.1145/3314221.3314584}
\showDOI{\tempurl}


\end{thebibliography}

\appendix

\end{document}